\documentclass[12pt]{article}

\usepackage{amsmath,amsfonts,amssymb,amsthm,textcomp,mathrsfs,mathrsfs,bbm}
\usepackage{braket,marvosym}
\usepackage{mathtools}
\usepackage[utf8]{inputenc}
\usepackage{graphicx}
\usepackage{color}
\usepackage[colorlinks=true,linkcolor=black,citecolor=black,plainpages=false,pdfpagelabels]{hyperref}
\usepackage{palatino}
\usepackage[T1]{fontenc}
\usepackage{fullpage}

\newtheorem{theorem}{Theorem}
\newtheorem{lemma}[theorem]{Lemma}
\newtheorem{corollary}[theorem]{Corollary}
\newtheorem{definition}{Definition}

\newcommand*{\cC}{\mathcal{C}}
\newcommand*{\cM}{\mathcal{M}}
\newcommand*{\cN}{\mathcal{N}}
\newcommand*{\bracket}[2]{\left\langle #1 \right| \left. #2 \right\rangle}
\newcommand*{\brackett}[3]{\left\langle #1 \right| \left. #2 \right. \left| #3 \right\rangle}
\newcommand*{\cD}{\mathcal{D}}
\newcommand*{\cX}{\mathcal{X}}
\newcommand*{\eye}{{\mathbbm{1}}}
\newcommand*{\mzero}{{\mathbf{0}}}

\newcommand*{\bbR}{{\mathbb{R}}}

\newcommand*{\cE}{\mathcal{E}}
\newcommand*{\cT}{{\mathcal{T}}}
\newcommand*{\cF}{\mathcal{F}}
\newcommand*{\cR}{\mathcal{R}}
\newcommand{\beq}{\begin{equation}}
\newcommand{\enq}{\end{equation}}
\newcommand{\ketbra}[1]{\ket{#1} \bra{#1}}

\newcommand{\tr}{\mathrm{Tr}}

\newcommand*{\cH}{\mathcal{H}}
\newcommand*{\cI}{\mathcal{I}}
\newcommand*{\cV}{\mathcal{V}}

\newcommand*{\renyi}{R\'{e}nyi }
\newcommand*{\bbU}{{\mathbb{U}}}

\newcommand{\Exp}{{\mathsf E}}

\newcommand{\ind}{{\mathrm{ind}}}

\DeclareMathOperator{\densitymatrix}{D}

\newcommand{\LL}{\mathrm{L}}
\DeclareMathOperator{\Herm}{Herm}
\DeclareMathOperator{\Pos}{Pos}

\newcommand{\mapone}{\cT_W}

\newcommand{\qmap}{{\mathcal{Q}}}
\newcommand{\err}{{\text{error}}}
\newcommand{\sold}{D^{\text{old}}}
\newcommand{\sand}{D^{\text{sand}}}
\newcommand{\qold}{Q^{\text{old}}}
\newcommand{\qsand}{Q^{\text{sand}}}
\newcommand{\csold}{H^{\text{old}}}
\newcommand{\csand}{H^{\text{sand}}}

\usepackage[colorlinks=true,linkcolor=black,citecolor=black,plainpages=false,pdfpagelabels]{hyperref}
\hypersetup{pdftitle={Template}}
\usepackage{color}

\begin{document}

\newcommand{\itwomax}{{^2I}_{\max}}

\title{Random coding exponents galore via decoupling}
\author{Naresh Sharma \\
Tata Institute of Fundamental Research \\
Mumbai 400005, India \\
Email: {\tt nsharma@tifr.res.in}}
\date{\today}
\maketitle

\begin{abstract}
A missing piece in quantum information theory,
with very few exceptions, has been to provide the random coding exponents for
quantum information-processing protocols. We remedy the situation by 
providing
these exponents for a variety of protocols including those at the top of the 
family tree of protocols.
Our line of attack is to provide an exponential bound on the decoupling error
for a restricted class of completely positive maps where a 
key term in the exponent is in terms of a \renyi $\alpha$-information-theoretic 
quantity for any $\alpha \in (1,2]$.
Among the protocols covered are fully quantum Slepian-Wolf, quantum state 
merging, quantum
state redistribution, quantum/classical
communication across channels with side
information at the transmitter with or without entanglement assistance,
and quantum communication across broadcast channels.
\end{abstract}

\section{Introduction}

Analysis of optimal resources needed/generated in an information-processing protocol
is one of the holy grails of information theory \cite{covertom, nielsen-chuang, hayashi, petz-book,
wilde-book}. Nice answers in terms of information-theoretic quantities
are obtained, in general, for large copies such as of inputs and channel uses. One part in establishing
these answers is the achievability that says that for resources arbitrarily close
to the optimal, there exists a protocol accomplishing the task with arbitrarily small error.

Achievability proofs come in various flavors and we list some of them
but not in the chronological order.
One way is via the law of large numbers (or
typicality) that involves making statements for large copies. Another way is via the smooth
information-theoretic quantities that are defined in terms of a semi-definite program
(see Refs. \cite{renner-thesis, tomamichel-thesis} and references therein). This method has the
advantage that one can make statements for any number of copies and it matches the
optimal answer for large number of copies using the law of large numbers. A third way has been
via the random coding exponents, i.e., one makes statements for any number of copies
by obtaining an exponential bound on the error of the protocol. In many comparisons with the second
method, this method provides stronger bounds and was pioneered by
Gallager who obtained such bounds for the classical capacity
\cite{gallager-expo-1965, gallager-68-book}.
Yet another method has been via the optimal terms in the asymptotic expansions of the
rate at which the resources are generated or used and this was pioneered
by Strassen \cite{strassen-1962}.

It is the Gallager's approach that would be further investigated in this paper. If one scours the literature on
the random coding exponents for quantum protocols, one finds that not much work has been done on this topic.
Indeed, apart from Burnashev and Holevo \cite{burnashev-1998}, Holevo \cite{holevo-2000}, and
Hayashi \cite{hayashi-2006}, no other work, to the best of the author's knowledge, provides
random coding exponents for the quantum protocols. (Exponential bounds on the error for the
Schumacher compression can be obtained without much difficulty leveraging the analysis
for the classical source compression \cite{hayashi}.) Burnashev and Holevo \cite{burnashev-1998} provide
the reliability function (loosely defined as the best exponent one could get for
large number of copies \cite{gallager-68-book}) for sending classical information across the quantum
channel for the case of pure states, and Holevo \cite{holevo-2000} extends it for the case of
commuting density matrices. Hayashi provides a random coding exponent
for the same protocol for general density matrices but his exponent when specialized to
classical does not match with Gallager's \cite{holevo-2000,gallager-expo-1965}.

Quantum information theory is much richer than the classical and with a plethora of protocols
(one can just glance at the family tree of quantum protocols
\cite{devetak-family-protocols-2004,devetak-resource-2008} to appreciate this), it is not just
important to provide the random coding exponents but, if possible, also a unified approach to
get these exponents for a variety of protocols.

Where would such a unified approach come from? An answer lies in decoupling, a phenomenon
where random evolution of a part of the quantum system would, on the average,
make it decouple from the other part.
That decoupling would be useful for quantum error correction was first observed by Schumacher
and Nielsen \cite{schumacher-nielsen-1996}. It has subsequently been recognized as a
building block in quantum information theory (see Refs. \cite{dupuis-thesis,hayden-qip-2011}
and references therein).

The decoupling theorem quantifies the average error between the state, part of which is
randomly evolved, and the
completely decoupled state, and is now known in various versions. We go through some of
them not necessarily in the chronological order. The one
provided by Hayden \emph{et al} \cite{hayden-et-al-decoupling-2008}
gives a bound in terms of dimensions of the quantum
systems involved and this, with an appeal to typicality for large copies, yields the optimal answers
$\relbar$ similar approach is followed by Abeyesinghe \emph{et al} \cite{fqsw-2009}.
Dupuis \emph{et al} provide another version that gives a bound in terms of
smooth entropies \cite{dupuis-et-al-2010}.

Another version by Dupuis gives an exponential bound for any number of copies and the exponent has two
\renyi $2$-conditional entropies: first one is computed using the density matrix that is evolved 
and the second one is computed using the Choi-Jamio{\l}kowski representation of a map \cite{dupuis-thesis}.

Since this version gives an exponential bound, it seems close to the stated purpose of this
paper but it is not quite there simply because for the random coding exponents, we shall need
the first term to be in terms of \renyi $\alpha$-conditional entropies for $\alpha$ arbitrarily
close to $1$. It is not necessary to strengthen the second term that determines the rate.

Could there be a way of modifying Dupuis' bound? This paper stems from asking this
question, answers it in the affirmative, and then applies the new version to obtain the random
coding exponents for a variety of protocols. In
particular, we are able to replace the first term by a \renyi $\alpha$-conditional
entropy for all $\alpha \in (1,2]$ (although adding some inconsequential terms in the process).
We do this by leveraging ideas from the independent works of Dupuis and Hayashi
\cite{dupuis-thesis,hayashi}.

Some of the protocols we analyze are at the top of the family tree of
protocols and the author didn't encounter any protocol that could be analyzed by other
versions of the decoupling theorem but not from the version provided in this paper. For the
protocols analyzed, the application
of our version of the decoupling theorem is, in some cases, but not always,
inspired by the application of other versions of the decoupling theorem.

We don't address how close the exponent in the proposed bounds might be to the reliability function.
There is, however, one resemblance between the exponents we
obtain and the reliability function for the classical case (in certain regimes), which is that
in both the cases, it is in terms of \renyi $\alpha$-information-theoretic quantities. 

The structure of the paper is as follows. Section \ref{prelims} provides the notation and definitions
used throughout this paper. Section \ref{dec-theorem} provides a new version of the decoupling theorem.
(There is a more general version provided as well in Appendix \ref{appendix3} although we don't use it!)
The subsequent sections apply this version to various protocols. Following protocols
are analyzed: Schumacher compression, fully quantum Slepian-Wolf, fully quantum reverse Shannon,
quantum state merging, quantum/classical communication across channels
with side information at the transmitter with or without
entanglement assistance, entanglement-assisted classical communication, quantum state
redistribution, quantum communication across broadcast channels, and destroying correlations by
adding classical randomness. The lemmas are provided in the appendix so as to not interrupt the flow.

\section{Notation and Preliminaries}
\label{prelims}

Let $\cH_A$ be the Hilbert space associated with the quantum system $A$.
We shall confine ourselves to the
finite dimensional Hilbert spaces in this paper and $|A|$ denotes the dimension of $\cH_A$.
$A \cong B$ implies that $|A| = |B|$.
For a system $A$, we denote $A^n$ to be a quantum system described by $\bigotimes_{i=1}^n \cH_{A_i}$,
where $A_i \cong A$, $i = 1,...,n$.
Let $\LL(\cH_A,\cH_B)$ be the set of all matrices from $\cH_A$ to $\cH_B$ and
$\LL(\cH_A)$ denotes $\LL(\cH_A,\cH_A)$. Let $\Herm(\cH_A)$,
$\Pos(\cH_A)$ $\subseteq \LL(\cH_A)$ be the set of Hermitian and positive semidefinite matrices
respectively. Let $\densitymatrix(\cH_A) \subseteq \Pos(\cH_A)$
be the set of unit trace matrices and $\densitymatrix_{\leqslant}(\cH_A)$ $\subseteq \Pos(\cH_A)$
be the set of matrices with trace not greater than $1$.
Let $\nu_{\sigma^A}$ be the number of distinct eigenvalues of $\sigma^A \in \Herm(\cH_A)$. For
$\rho^A, \sigma^A \in \Herm(\cH_A)$, let $\{ \rho^A \geqslant \sigma^A \}$ denote the
projector onto the subspace spanned by the eigenvectors corresponding to the non-negative
eigenvalues of $\rho^A - \sigma^A$. Let $X \cdot \rho \equiv X \rho X^\dagger$.
For $X \in \LL(\cH_A, \cH_B)$ (also denoted as $X^{A \to B}$),
the trace norm, $\|X\|_1$, is the sum of its singular values.
The Fidelity between $\rho, \sigma \in \Pos(\cH_A)$ is
$F(\rho,\sigma) \equiv \| \sqrt{\rho} \sqrt{\sigma} \|_1$.

Let $\bbU(A)$ be a Unitary $2$-design on a quantum system $A$
(see Ref. \cite{dupuis-thesis} and references therein).
For a function $f: \bbU(A) \to \LL(\cH_E)$, $\Exp_U f(U)$ denotes the
expectation taken over a random Unitary $U$ distributed uniformly on $\bbU(A)$.

Let $\ket{\Phi}^{A A^\prime}$ be the maximally entangled state (MES)
on $A A^\prime$, i.e., for $A \cong A^\prime$,
orthonormal bases $\{ \ket{i}^A \}$ and $\{ \ket{i}^{A^\prime} \}$,
$\ket{\Phi}^{A A^\prime} \equiv |A|^{-1/2} \sum_{i=1}^{|A|} \ket{i}^A \ket{i}^{A^\prime}$.
Let the maximally mixed state in $\cH_A$ be denoted by $\pi^A \equiv \eye^A/|A|$, where $\eye^A$
is the Identity matrix. The zero matrix (with all entries as zero) is denoted by $\mzero$.

A matrix $V^{A \to B}$ is an isometry if either
$V^\dag V = \eye$ or $V V^\dag = \eye$, and is a partial isometry if its singular values are
either $0$ or $1$. A full-rank partial isometry $V^{A \to B}$ has rank $\min\{|A|,|B|\}$.

The Kronecker delta function is $\delta_{j,k} = 1$ if $j = k$, and $0$ otherwise. The indicator
function $\ind_{\text{condition}}$ $= 1$ if condition is true, and $0$ otherwise.
The partial trace over $B$ of $\rho^{AB} \in \LL(\cH_A \otimes \cH_B)$ is denoted by either $\tr_B \rho^{AB}$
or $\rho^A$. For a pure state $\ket{\Psi}^{AB}$, $\Psi^{AB} = \ketbra{\Psi}^{AB}$,
and it does not necessarily imply that $\Psi^A$ is also a pure state.
All the logarithms in this paper are to the base $2$ and $\exp(x)$ denotes $2^x$, $x \in \bbR$.
We define $\Xi(\varepsilon) \equiv \sqrt{\varepsilon (2 + \varepsilon + 2 \sqrt{1+
\varepsilon})}$ for $\varepsilon \geqslant 0$.

With an abuse of notation, we call a weighted sum of exponentially
decaying terms also as exponential decay, i.e., for $x$, $\alpha_i$,
$\beta_i$ > 0, $i=1,...,n$, $n$ finite, we call $\sum_{i=1}^n \beta_i
\exp\{ -\alpha_i x \}$ as exponentially decaying with $x$. All the error bounds
that we provide in this paper can be put in this form.

\subsection{Super-operators}
\label{subsec-so}

A super-operator $\cT^{A \to B}$ is a map from $\LL(\cH_A) \to \LL(\cH_B)$. Important classes
include completely positive maps $\cT^{A \to B}$, which map $\Pos(\cH_A \otimes \cH_R)$
to $\Pos(\cH_B \otimes \cH_R)$ for any ancilla $R$, and completely positive and trace preserving (cptp)
maps which are completely positive and have an additional property that the trace is preserved.

The Choi-Jamio{\l}kowski representation of a map $\cT^{A \to E}$ is given by
$\omega^{E A^\prime}_{\cT}$ $\equiv$ $\cT^{A \to E}(\Phi^{A A^\prime})$. To a completely positive map,
we associate a quantity $\Theta(\cT)$ defined as the negative of the \renyi old $2$-conditional
entropy (defined in Section \ref{info-quant}) and is given by
\beq
\label{yae8}
\Theta(\cT) \equiv - \csold_2(A^\prime | E)_{\omega_{\cT}^{E A^\prime}}.
\enq

Concatenation of two maps, i.e., $\cE$ followed by $\cD$ is denoted by $\cD \circ \cE$,
and with a slight abuse of notation, for an isometry $V$ and a map $\cE$, $\cE \circ V (\rho)$ denotes
$\cE(V \cdot \rho)$, and $V \circ \cE (\rho)$ denotes $V \cdot \cE(\rho)$.

We now define three maps.
For $\sigma^{AB} \in \LL(\cH_A \otimes \cH_B)$,
$\qmap_A(\sigma^{AB}) \equiv |A| \tr_A \sigma^{AB} (\sigma^{AB})^\dag - \sigma^B (\sigma^B)^\dag$.
For $\rho, \sigma \in \Pos(\cH_A)$,
the spectral decomposition $\sigma = \sum_{i=1}^{\nu_{\sigma}} \lambda_i \Pi_i$,
where $\lambda_i$'s are all distinct and $\Pi_i$'s are projectors,
a pinching map in the eigenbasis of $\sigma$ is defined as
$\cM_\sigma(\rho) \equiv \sum_{i=1}^{\nu_{\sigma}} \Pi_i \rho \Pi_i$.
Let $W^{A \to B}$, $|B| \leqslant |A|$, be a full-rank partial isometry.
Then a compressive map $\cC_W^{A \to B}$ is defined as
$\cC_W(\rho^A) \equiv W \rho^A W^\dagger + \left[ \tr (\eye^A - W^\dagger W) \rho^A \right] \pi^B$.

\begin{definition}[Class-$1$ maps]
A map $\cT^{A \to E}$ is said to be in class-$1$ if it is completely positive and
for any $\sigma \in \LL(\cH_A)$, $\Exp_U \| \cT(U \cdot \sigma) \|_1$ $\leqslant$ $\| \sigma \|_1$.
\end{definition}

Note that all cptp maps fall under class-$1$. Another set of
completely positive maps under class-$1$ are those with
$\tr \, \cT(\eye^A) = |A|$ (see Lemma \ref{yal6} for proof). An example of such a map (taken from
Ref. \cite{dupuis-thesis}) that we shall use later in the paper is given by
\beq
\label{yae31}
\mapone^{A \to B}(\sigma^A) \equiv \frac{|A|}{|B|} (W^{A \to B} \cdot \sigma^A),
\enq
where $W^{A \to B}$, $|A| \geqslant |B|$, is a full-rank partial isometry.

\subsection{Information-theoretic quantities}

\label{info-quant}

The quantum relative entropy from $\rho$ to $\sigma$ is given by
$D(\rho \| \sigma) \equiv \tr \rho (\log \rho - \log \sigma)$, the von Neumann
entropy of $\rho^A \in \densitymatrix(\cH_A)$
is given by $H(A)_\rho \equiv - \tr \rho^A \log \rho^A$.
For a tripartite state $\rho^{ABC}$, the conditional entropy of $A$ given $B$
is given by $H(A|B)_\rho \equiv H(AB)_\rho - H(B)_\rho$, the
conditional mutual information between $A$ and $B$ given $C$ is
$I(A:B|C)_\rho \equiv H(A|C)_\rho - H(A|BC)_\rho$, and the coherent information is given
by $I(A \rangle B)_\rho \equiv -H(A|B)_\rho$.
The \renyi generalizations of the quantum relative entropy can be
done in various ways and we mention two prominent candidates.

\begin{definition}[\renyi entropies]
For $\alpha \in (0,2] \backslash \{1\}$, from $\rho$ to $\sigma$,
the quasi old $\alpha$-relative entropy is given by
$\qold_\alpha(\rho \| \sigma) \equiv \tr \rho^\alpha \sigma^{1-\alpha}$,
and the quasi sandwiched $\alpha$-relative entropy (proposed
independently in Refs. \cite{lennert-2013,wilde-2013}) is given by
$\qsand_\alpha(\rho \| \sigma) \equiv \tr \left( \sigma^{\frac{1-\alpha}{2 \alpha}}
\rho \sigma^{\frac{1-\alpha}{2 \alpha}} \right)^\alpha$.
The \renyi old (sandwiched) $\alpha$-relative entropy from $\rho$ to $\sigma$ is given by
\beq
D^{\text{old (sand)}}_\alpha(\rho \| \sigma) \equiv \frac{1}{\alpha-1} \log
Q^{\text{old (sand)}}_\alpha(\rho \| \sigma), ~~
\alpha \in (0,2] \backslash \{1\}.
\enq
%Let $H_{\alpha}(A)_\rho$, $\alpha \geqslant 0$, be the $\alpha$-entropy of
%$\rho^A$ given by
%\beq
%H_{\alpha}(A)_\rho \equiv \frac{1}{1-\alpha} \log \tr (\rho^A)^{\alpha}.
%\enq
We can extend these definitions to include $\alpha = 1$ by taking limits and we drop the subscript
and the superscript.
The \renyi $\alpha$-conditional entropies of $A$ given $B$ are defined as
\begin{align}
H^{\text{type}}_\alpha(A|B)_\rho & \equiv -\inf_{\sigma^B \in \densitymatrix(\cH_B)}
D^{\text{type}}_\alpha(\rho^{AB} \| \eye^A \otimes \sigma^B) \\
^{\downarrow}H^{\text{type}}_\alpha(A|B)_\rho & \equiv
- D^{\text{type}}_\alpha(\rho^{AB} \| \eye^A \otimes \rho^B),
\end{align}
where `type' is `old' or `sand'.
\end{definition}

It follows from Refs. \cite{petz-quasi-entr-1986, petz-quasi-entr-2010, lennert-2013, wilde-2013}
that for $\alpha \in (0,2] \backslash \{1\}$ and a cptp map $\cE$,
$D^{\text{type}}_\alpha(\rho \| \sigma) \geqslant D^{\text{type}}_\alpha\left[ \cE(\rho) \| \cE(\sigma)\right]$.

There are duality relations known for a tripartite pure state $\Psi^{ABC}$ . One such example is
$\csand_\alpha(A|B)_\Psi + \csand_{\widetilde{\alpha}}(A|C)_\Psi = 0$,
$\widetilde{\alpha} = 1/\alpha$, $\alpha \in [0.5,1) \cup (1,2]$. See Ref. \cite{tomamichel-duality-2013} and references therein for a
complete list of duality relations.
In the remainder of the paper, the `type' superscript is dropped, and it implies that the expression holds for
either one and one could pick one's favorite. For example,
$D_\alpha(\rho \| \sigma)$ denotes either
$\sold_\alpha(\rho \| \sigma)$ or $\sand_\alpha(\rho \| \sigma)$. Furthermore,
while invoking the above duality relations, since there are many options,
we also drop the downarrow superscript from the conditional entropies
and assume that appropriate superscript is implicitly assumed and $\widetilde{\alpha}$ is
assumed to be an appropriate function of $\alpha$ depending on the type of
conditional entropies involved.

\section{Yet another version of the decoupling theorem with a useful R\'{e}nyification}
\label{dec-theorem}

In this section, we provide a version of the decoupling theorem where the crucial term in
the exponent is in terms of a \renyi $\alpha$-information-theoretic quantity for $\alpha \in (1,2]$ instead of just
$\alpha = 2$ as provided in Ref. \cite{dupuis-thesis}.

We leverage ideas from the independent works of Dupuis and Hayashi and in particular
Theorem 3.7 in Ref. \cite{dupuis-thesis} and Lemma 9.2 in Ref. \cite{hayashi}.

\begin{theorem}
\label{theorem1}
Let $\rho^{AR} \in \densitymatrix(\cH_{AR})$ and $\cT^{A \to E}$ be a class-$1$ map. 
Then for $\alpha \in (1,2]$, a random Unitary $U$ acting on $A$, we have for any $\sigma^R$
$\in \densitymatrix(\cH_{R})$,
\begin{multline}
\Exp_{U} \left\| \cT(U \cdot \rho^{AR}) - \omega^E_\cT \otimes \rho^R \right\|_1 \\
\leqslant 4 \exp\left\{ \frac{\alpha-1}{2\alpha} \left[ \log \nu_{\sigma^R}
+ D_\alpha(\rho^{AR} \| \eye^A \otimes \sigma^R) + \Theta(\cT) \right] \right\}.
\end{multline}
In particular, for $n$ copies, a random Unitary $U$ acting on
$A^n$, and a class-$1$ map
$\cT^{A^n \to E}$, we have
\begin{multline}
\Exp_U \left\| \cT \left[ U \cdot (\rho^{AR})^{\otimes n} \right] -
\omega^E_\cT \otimes (\rho^R)^{\otimes n} \right\|_1 \\
\leqslant 4 \exp\left\{ \frac{\alpha-1}{2\alpha} \Big[ |R| \log (n+1)
- n H_\alpha(A|R)_\rho + \Theta(\cT) \Big] \right\}.
\end{multline}
\end{theorem}
\begin{proof}
For a $\zeta > 0$, let
$\Pi^{AR} \equiv \left\{ \cM_{\eye^A \otimes \sigma^R}(\rho^{AR}) \geqslant \zeta \eye^A \otimes \sigma^R
\right\}$, $\hat{\Pi}^{AR} \equiv \eye^{AR} - \Pi^{AR}$,
$\mu_1 \equiv \omega^{E}_\cT \otimes \tr_A \left\{ \Pi^{AR} \rho^{AR} \right\}$, and
$\mu_2 \equiv \omega^{E}_\cT \otimes \tr_A \left\{ \hat{\Pi}^{AR} \rho^{AR} \right\}$.
Note that $\mu_1 + \mu_2 = \omega^{E}_\cT \otimes \rho^R$. We now have
\begin{multline}
\Exp_{U} \left\| \cT(U \cdot \rho^{AR}) - \omega^E_\cT \otimes \rho^R \right\|_1
= \Exp_{U} \left\| \cT \left[ U \cdot (\Pi^{AR} \rho^{AR}) \right] - \mu_1 +
\cT \left[ U \cdot (\hat{\Pi}^{AR} \rho^{AR})  \right] - \mu_2 \right\|_1 \\
\leqslant \Exp_{U} \left\| \cT \left[ U \cdot (\Pi^{AR} \rho^{AR}) \right] - \mu_1 \right\|_1 +
\Exp_{U} \left\| \cT \left[ U \cdot (\hat{\Pi}^{AR} \rho^{AR}) \right] - \mu_2 \right\|_1,
\end{multline}
where we have used the triangle inequality.

We attack the first term.
\begin{align}
\Exp_{U} \left\|
\cT \left[ U \cdot (\Pi^{AR} \rho^{AR})  \right] - \mu_1 \right\|_1
& \leqslant \Exp_{U} \left\| \cT \left[ U \cdot (\Pi^{AR} \rho^{AR})  \right] \right\|_1 + \| \mu_1 \|_1 \\
& \leqslant 2 \, \Exp_{U} \left\| \cT \left[ U \cdot (\Pi^{AR} \rho^{AR})  \right] \right\|_1 \\
& \leqslant 2 \, \left\| \Pi^{AR} \rho^{AR} \right\|_1 \\
& \leqslant 2 \zeta^{\frac{1-\alpha}{2}} \, \exp \left\{ \frac{\alpha - 1}{2}
D_\alpha(\rho^{AR} \| \eye^A \otimes \sigma^R) \right\},
\end{align}
where the first inequality follows from the triangle inequality,
the second inequality follows from the convexity of the trace norm to have
\begin{align}
\| \mu_1 \|_1 = \left\| \Exp_U  \left\{ \cT \left[ U \cdot (\Pi^{AR} \rho^{AR})  \right] \right\} \right\|_1
\leqslant \Exp_U \left\| \cT \left[ U \cdot (\Pi^{AR} \rho^{AR})  \right] \right\|_1,
\end{align}
the third inequality follows from the definition of class-$1$ maps,
the fourth inequality follows from Lemma \ref{yal3} (proved by Hayashi \cite{hayashi}).

We now attack the second term. Let $\Delta_U \equiv \cT \left[ U \cdot (\hat{\Pi}^{AR} \rho^{AR}) \right] - \mu_2$,
and $\theta^E$ $\in \densitymatrix(\cH_E)$ be such that
$\Theta(\cT) = - \sold_2( \omega^{E A^\prime}_\cT \| \theta^E \otimes \eye^{A^\prime})$.
We now have
\begin{align}
%\Exp_U & \Big\| \cT \left[ U \cdot (\hat{\Pi}^{AR} \rho^{AR}) \right] - \mu_2 \Big\|_1 \nonumber \\
\Exp_U \left\| \Delta_U^{ER} \right\|_1
& \leqslant \Exp_U \sqrt{ \tr \left[ (\theta^{E})^{-1} \otimes (\sigma^R)^{-1} \right]
\Delta_U \Delta_U^\dag } \\
& \leqslant \sqrt{ \tr \left[ (\theta^{E})^{-1} \otimes (\sigma^R)^{-1} \right]
\Exp_U \left\{ \Delta_U \Delta_U^\dag \right\} } \\
& \leqslant \sqrt{ \frac{|A|^2  \tr \left\{ (\theta^{E})^{-1} \tr_{A^\prime} \left( \omega^{E A^\prime}_\cT 
\right)^2 \right\}
\tr \left\{ (\sigma^R)^{-1} \tr_A \, \hat{\Pi}^{AR} (\rho^{AR})^2 \hat{\Pi}^{AR} \right\}}{|A|^2-1} } \\
%& = \sqrt{ \frac{|A|^2  \exp \left\{ \Theta(\cT) \right\}
%\tr \left\{ (\sigma^R)^{-1} \tr_A \, \hat{\Pi}^{AR} (\rho^{AR})^2 \hat{\Pi}^{AR} \right\}}{|A|^2-1} } \\
& \leqslant
\sqrt{ \frac{\nu_{\sigma^R} \zeta |A|^2 \exp \left\{ \Theta(\cT) \right\}}{|A|^2-1} },
\end{align}
where the first inequality follows since for any matrix $\Upsilon$ and a density matrix $\kappa$ (with
appropriate dimensions),
$\| \Upsilon \|_1$ $\leq$ $\sqrt{\tr \kappa^{-1} \Upsilon \Upsilon^\dagger}$, the second inequality follows
from the concavity of $x \to \sqrt{x}$, the third inequality follows from Lemma \ref{yal9},
and the last inequality follows from Lemma \ref{yal4} (proved by Hayashi \cite{hayashi}).
We now have
\begin{multline}
\Exp_U \left\| \cT(U \cdot \rho^{AR}) - \omega^E_\cT \otimes \rho^R \right\|_1 \\
\leqslant 2 \zeta^{\frac{1-\alpha}{2}} \, \exp \left\{ \frac{\alpha - 1}{2}
D_\alpha(\rho^{AR} \| \eye^A \otimes \sigma^R) \right\}
+ \sqrt{ \frac{\nu_{\sigma^R} \zeta |A|^2 \exp \left\{ \Theta(\cT) \right\}}{|A|^2-1} }.
\end{multline}

Note that $\zeta$ is a free parameter and a convenient upper bound for
\beq
\min_\zeta \left(  x \zeta^{\frac{1-\alpha}{2}} + y \zeta^{1/2} \right)
\enq
is obtained by choosing $\zeta = (x y^{-1})^{\frac{2}{\alpha}}$. Making that choice by feeding in appropriate 
values of $x$, $y$, taking $|A|^2/(|A|^2-1) \leqslant 4/3$, noting that for $\alpha \in (1,2]$,
$2^{1/\alpha} (4/3)^{(\alpha-1)/(2\alpha)} \leqslant 2$, we get
\begin{multline}
\label{yae43}
\Exp_U \left\| \cT(U \cdot \rho^{AR}) - \omega^E_\cT \otimes \rho^R \right\|_1 \\
\leqslant 4 \exp\left\{ \frac{\alpha-1}{2\alpha} \Big[ \log \nu_{\sigma^R}
+ D_\alpha(\rho^{AR} \| \eye^A \otimes \sigma^R) + \Theta(\cT) \Big] \right\}.
\end{multline}
Note that this is a convenient upper bound and while one could further optimize the choice
of $\zeta$, for $\alpha$ near $1$, the above bound is near the optimal.

For $n$ copies, a random Unitary over $A^n$, and a class-$1$ map $\cT^{A^n \to E}$, using
\eqref{yae43}, we have
\begin{multline}
\Exp_U \left\| \cT \left[ U \cdot (\rho^{AR})^{\otimes n} \right] -
\omega^E_\cT \otimes (\rho^R)^{\otimes n} \right\|_1 \\
\leqslant 4 \min_{\sigma^R} \exp\left\{ \frac{\alpha-1}{2\alpha} \Big[ \nu_{(\sigma^R)^{\otimes n}}
+ D_\alpha \left[ (\rho^{AR})^{\otimes n} || \eye^{A^n} \otimes (\sigma^R)^{\otimes n} \right] +
\Theta(\cT) \Big] \right\} \\
\leqslant 4 \exp\left\{ \frac{\alpha-1}{2\alpha} \Big[ |R| \log (n+1)
- n H_\alpha(A|R)_\rho + \Theta(\cT) \Big] \right\}, \hspace{1.26in}
\end{multline}
where the first inequality follows from \eqref{yae43} and making a (possibly suboptimal) choice
of a product state, and the second inequality follows since
we have used a convenient upper bound that for any
$\sigma^R \in \densitymatrix(\cH_R)$,
$\log \nu_{(\sigma^R)^{\otimes n}} \leqslant |R| \log(n+1)$ (see
Theorem 11.1.1 in Ref. \cite{covertom} or Lemma 3.7 in Ref. \cite{hayashi}) and we choose
$\sigma^R$ to be the one that minimizes $D_\alpha(\rho^{AR} \| \eye^A \otimes \sigma^R)$. QED.
\end{proof}

We now have the following corollary of Theorem 1.
\begin{corollary}
\label{corollary1}
For $i = 1,...,K$, let $\cT_i^{A^n \to E_i}$ be class-$1$ maps, and
$\rho^{AR_i}_i \in \densitymatrix(\cH_{AR_i})$. Then there exists a Unitary $U$ over $A^n$
such that for all $i=1,...,K$, and $n \in \mathbb{N}$,
\begin{multline}
\left\| \cT_i \left[ U \cdot (\rho^{AR_i}_i)^{\otimes n} \right] -
\omega^{E_i}_{\cT_i} \otimes (\rho^{R_i}_i)^{\otimes n} \right\|_1 \\
\leqslant 4 K \exp\left\{ \frac{\alpha-1}{2\alpha} \Big[ |R_i| \log (n+1)
- n H_\alpha(A|R_i)_{\rho_i} + \Theta(\cT_i) \Big] \right\}.
\end{multline}
\end{corollary}
\begin{proof}
It follows from Theorem \ref{theorem1} that for all $i=1,...,K$,
\begin{multline}
\Exp_U \left\| \cT_i \left[ U \cdot (\rho^{AR_i}_i)^{\otimes n} \right] -
\omega^{E_i}_{\cT_i} \otimes (\rho^{R_i}_i)^{\otimes n} \right\|_1 \\
\leqslant 4 \exp\left\{ \frac{\alpha-1}{2\alpha} \Big[ |R_i| \log (n+1)
- n H_\alpha(A|R_i)_{\rho_i} + \Theta(\cT_i) \Big] \right\}.
\end{multline}
We now invoke Lemma I.7 in Ref. \cite{dupuis-thesis} to arrive at the claim. (Note that
Lemma I.7 in Ref. \cite{dupuis-thesis} stipulates a multiplying factor of $K+1$ instead of $K$ but
it can be easily strengthened.)
\end{proof}

It is possible to provide a unified theorem that yields both
Theorem \ref{theorem1} and Lemma 9.2 in Ref. \cite{hayashi}
as special cases. We do that in Theorem \ref{theorem1-2} (see Appendix \ref{appendix3}) and we
note that although we provide this unified theorem,
we don't use it for the protocols treated later in the paper!

\section{Schumacher compression}
\label{schu-section}

\begin{definition}
A $(\rho, \err, n)$ {\bf Schumacher compression} protocol consists of
$n$ copies of $\rho^{A}$ (with a purification $\Psi^{AR}$), Alice applying
an encoding cptp map $\cE : A^n \to B$, and Bob applying a decoding cptp map $\cD : B \to \widetilde{A}^n$
such that for
$\rho^{\widetilde{A}^n R^n} \equiv \cD^{B \to \widetilde{A}^n} \circ \cE^{A^n \to B} \left[ (\Psi^{AR})^{\otimes n}
\right]$,
\beq
\left\| \rho^{\widetilde{A}^n R^n} - (\Psi^{\widetilde{A}R})^{\otimes n} \right\|_1 \leqslant \err.
\enq
$(\log |B|)/n$ is called the {\bf compression rate} of the protocol. A real number $\cR_C$ is called an
{\bf achievable rate} if there exist, for $n \to \infty$, Schumacher
compression protocols with compression rate approaching $\cR_C$ and the error
approaching $0$.
\end{definition}

\begin{theorem}[Schumacher, 1995 \cite{schu-noiseless-1995}]
The smallest achievable rate for Schumacher compression is given
by $H(A)_\rho$.
\end{theorem}

We prove the following theorem.
\begin{theorem}
For any $n \in \mathbb{N}$, there exists a $(\rho, \err, n)$ Schumacher compression protocol such that
for any $\delta > 0$,
\beq
\frac{\log|B|}{n} = |R| \frac{\log(n+1)}{n} + H_{\widetilde{\alpha}}(A)_{\Psi} + \delta.
\enq
and the error approaches $0$ exponentially in $n$.
\end{theorem}
\begin{proof}
Consider a full-rank partial isometry $W^{A^n \to B}$, $|B| \leqslant |A|^n$. Then,
using Theorem \ref{theorem1}, there exists a Unitary $U$ over $A^n$,
\begin{multline}
\left\| \tr_B \circ \cT_W^{A^n \to B} \left[ U \cdot (\Psi^{AR})^{\otimes n} \right] - (\Psi^R)^{\otimes n} \right\|_1 \\
\leqslant 4 \exp\left\{ \frac{\alpha-1}{2\alpha} \Big[ |R| \log(n+1) - n H_{\alpha}(A | R)_{\Psi}
+ \Theta(\tr_B \circ \cT_W) \Big] \right\} \\
= 4 \exp\left\{ \frac{\alpha-1}{2\alpha} \Big[ |R| \log(n+1) + n H_{\widetilde{\alpha}}(A)_{\Psi}
- \log|B| \Big] \right\} \equiv \varepsilon_n,
\end{multline}
where we have used $\Theta(\tr_B \circ \cT_W) = -\log |B|$ from Lemma \ref{yal5}.
We claim using Lemma \ref{yal11} that there exists a Unitary $V^{A^n \to A^n}$ such that
\beq
\label{yae20}
\left\| W^\dag \cdot \cT_W^{A^n \to B} \left[ U \cdot (\Psi^{AR})^{\otimes n} \right] -
V \cdot (\Psi^{AR})^{\otimes n} \right\|_1 \leqslant \Xi(\varepsilon_n),
\enq
and hence, using monotonicity of the trace norm under cptp maps ($\cC_W$ in this case),
\beq
\left\| \cT_W^{A^n \to B} \left[ U \cdot (\Psi^{AR})^{\otimes n} \right] - \cC_W \left[
V \cdot (\Psi^{AR})^{\otimes n} \right] \right\|_1
\leqslant \Xi(\varepsilon_n),
\enq
or
\beq
\label{yae21}
\left\| W^\dag \cdot \cC_W \left[ V \cdot (\Psi^{AR})^{\otimes n} \right] -
W^\dag \cdot \cT^{A^n \to B} \left[ U \cdot (\Psi^{AR})^{\otimes n} \right] \right\|_1
\leqslant \Xi(\varepsilon_n).
\enq
Define a partial isometry $W_2^{A^n \to B}$ as
$W_2 \equiv W V$ and note that $\cC_{W_2}(\sigma^{AR}) = \cC_W (V \cdot \sigma^{AR})$.
We now have
\begin{align}
\Big\| W_2^\dag \cdot & \cC_{W_2} \left[ (\Psi^{AR})^{\otimes n} \right] -
(\Psi^{AR})^{\otimes n} \Big\|_1 \nonumber \\
& = \left\| V^\dag W^\dag \cdot \cC_W \left[ V \cdot (\Psi^{AR})^{\otimes n} 
\right] - (\Psi^{AR})^{\otimes n} \right\|_1  \\
& = \left\| W^\dag \cdot \cC_W \left[ V \cdot (\Psi^{AR})^{\otimes n} \right] -
V \cdot (\Psi^{AR})^{\otimes n} \right\|_1 \\
& \leqslant \left\| W^\dag \cdot \cC_W \left[ V \cdot (\Psi^{AR})^{\otimes n} \right]
- W^\dag \cdot \cT^{A^n \to B} \left[ U \cdot (\Psi^{AR})^{\otimes n} \right]
\right\|_1 + \nonumber \\
& \hspace{1in} \left\| W^\dag \cdot \cT^{A^n \to B} \left[ U \cdot 
(\Psi^{AR})^{\otimes n} \right] - V \cdot (\Psi^{AR})^{\otimes n} \right\|_1  \\
& \leqslant 2 \, \Xi(\varepsilon_n),
\end{align}
where we have used the triangle inequality, \eqref{yae20}, and \eqref{yae21}.
It is now clear that Alice applies $\cC_{W_2}$ and Bob applies the isometry $W_2^\dag$.
The claim now follows readily.
\end{proof}

\noindent {\bf Remark}: This is not the best exponent for this protocol and one can get
the exponent that matches with the classical case (see Prob. 2.15 in Ref. \cite{csiszar-korner-book})
when specialized and this can be construed from the treatment in Ref. \cite{hayashi}.
Our purpose of stating the above proof is that the ideas would prove useful
for other protocols later in this paper since it is based on decoupling.

\section{Fully quantum Slepian-Wolf (FQSW)}
\label{fqsw}

\begin{definition}
A $(\Psi, \err, n)$ FQSW protocol consists of
$n$ copies of a pure state $\ket{\Psi}^{ABR}$ shared between with Alice
($A$) and Bob ($B$), and reference system ($R$), Alice applying an encoding cptp map
$\cE: A^n \to A_1 A_2$, quantum communication across a noiseless quantum
channel from Alice to Bob $\cI^{A_2 \to B_2}$, and Bob applying a decoding cptp map
$\cD: B_2 B^n \to B_1 \widetilde{B}_3^n B_3^n$ such that for
\beq
\rho^{A_1 B_1 \widetilde{B}_3^n B_3^n R^n} \equiv \cD^{B_2 B^n \to B_1 \widetilde{B}_3^n B_3^n}
\circ \cI^{A_2 \to B_2} \circ \cE^{A^n \to A_1 A_2} [(\Psi^{ABR})^{\otimes n}],
\enq
\beq
\left\| \rho^{A_1 B_1 \widetilde{B}_3^n B_3^n R^n} - \Phi^{A_1 B_1} \otimes
(\Psi^{\widetilde{B}_3 B_3 R})^{\otimes n} \right\|_1 \leqslant \err.
\enq
The number $(\log |A_2|)/n$ is called the {\bf quantum communication rate}
and $(\log |A_1|)/n$ is called the {\bf entanglement gain rate} of the protocol.

A pair of real numbers $(\cR_Q,\cR_E)$ is called an {\bf achievable rate pair} if there
exist, for $n \to \infty$, FQSW protocols with quantum communication rate approaching
$\cR_Q$, entanglement gain rate approaching $\cR_E$, and $\err$ approaching $0$.
\end{definition}

The achievable rates are described by the following theorem.

\begin{theorem}[Abeyesinghe \emph{et al}, 2009 \cite{fqsw-2009}]
\label{yatheorem1}
The following rates are achievable for the FQSW:
\beq
\cR_Q > \frac{1}{2} I(A:R)_\Psi \hspace{0.5in} \text{and} \hspace{0.5in}
\cR_E < \cR_Q + H(A|R)_\Psi.
\enq
\end{theorem}

Our goal in the remainder of this section is to provide the achievability of the above
rate region with error decaying to $0$ exponentially in $n$.

\begin{theorem}
For any $n \in \mathbb{N}$, there exists a $(\Psi, \err, n)$
FQSW protocol for any $\alpha \in (1,2]$, and
$\delta_1, \delta_2 > 0$, such that
\begin{align}
\frac{\log|A_2|}{n} & = \frac{1}{2} \Big[ H_{\widetilde{\alpha}}(A)_\Psi - H_{\alpha}(A | R)_{\Psi} \Big] +
(|B|+1) |R| \frac{ \log(n+1) }{2n} + \frac{\delta_1 + \delta_2}{2}, \\
\frac{\log|A_1|}{n}  & = \frac{\log |A_2|}{n} + H_{\alpha}(A | R)_{\Psi} - |R| \frac{ \log(n+1) }{n} - \delta_2,
\end{align}
and the error approaches $0$ exponentially in $n$.
\end{theorem}
\begin{proof}
Let $W : A^n \to A_1 A_2$ be a full-rank partial isometry with $|A_1| |A_2| \leqslant |A|$.
Then, using Corollary \ref{corollary1}, we claim that there exists a Unitary $U$ over $A^n$ such that
for $\alpha \in (1,2]$,
\begin{multline}
\label{yae55}
\left\| \tr_{A_1 A_2} \circ \mapone^{A^n \to A_1 A_2} \left[ U \cdot (\Psi^{ABR})^{\otimes n} \right] -
(\Psi^{BR})^{\otimes n} \right\|_1 \\
\leqslant 8 \exp \left\{ \frac{\alpha-1}{2\alpha} \Big[ |B| |R| \log(n+1) - n H_{\alpha}(A | B R)_{\Psi}
+ \Theta( \tr_{A_1 A_2} \circ \mapone ) \Big] \right\} \\
= 8 \exp \left\{ \frac{\alpha-1}{2\alpha} \Big[ |B| |R| \log(n+1) + n H_{\widetilde{\alpha}}(A)_{\Psi}
-  \log|A_1| |A_2| \Big] \right\} \equiv \varepsilon_n,
\end{multline}
and
\begin{multline}
\label{yae19}
\left\| \tr_{A_2} \circ \mapone^{A^n \to A_1 A_2} \left[ U \cdot (\Psi^{AR})^{\otimes n} \right] - \pi^{A_1} \otimes
(\Psi^R)^{\otimes n} \right\|_1 \\
\leqslant 8 \exp \left\{ \frac{\alpha-1}{2\alpha} \left[ |R| \log(n+1) - n H_{\alpha}(A | R)_{\Psi}
+ \log \frac{|A_1|}{|A_2|} \right] \right\} \equiv \vartheta_n.
\end{multline}
It follows from \eqref{yae19} and Lemma \ref{yal11}
that there exists an isometry $\widetilde{U}^{A_2 B^n \to B_1 \widetilde{B}_3^n B_3^n}$ such that
\beq
\label{yae46}
\left\| \widetilde{U} \cdot \left\{ \mapone^{A^n \to A_1 A_2} \left[ U \cdot (\Psi^{ABR})^{\otimes n} \right]
\right\} - \Phi^{A_1 B_1} \otimes (\Psi^{\widetilde{B}_3^n B_3^n R})^{\otimes n} \right\|_1 \leqslant
\Xi( \vartheta_n ).
\enq
It follows from \eqref{yae55} and Lemma \ref{yal11} that there exists a Unitary $V^{A^n \to A^n}$
such that
\begin{multline}
\label{yae45}
\Xi( \varepsilon_n ) \geqslant
\left\| W^\dag \cdot \mapone^{A^n \to A_1 A_2} \left[ U \cdot (\Psi^{ABR})^{\otimes n} \right] -
V \cdot (\Psi^{ABR})^{\otimes n} \right\|_1 \\
\geqslant \left\| \mapone^{A^n \to A_1 A_2} \left[ U \cdot (\Psi^{ABR})^{\otimes n} \right] -
\cC_W \left[ V \cdot (\Psi^{ABR})^{\otimes n} \right] \right\|_1 \hspace{0.75in} \\
= \left\| \widetilde{U} \cdot \left\{ \mapone^{A^n \to A_1 A_2} \left[ U \cdot (\Psi^{ABR})^{\otimes n} \right]
\right\} - \widetilde{U} \cdot \left\{ \cC_W \left[ V \cdot (\Psi^{ABR})^{\otimes n} \right] \right\} \right\|_1,
\hspace{0.08in}
\end{multline}
where the second inequality follows using the monotonicity and noting that
\beq
\cC_W \left\{ W^\dag \cdot \mapone^{A^n \to A_1 A_2} \left[ U \cdot (\Psi^{ABR})^{\otimes n} \right] \right\}
= \mapone^{A^n \to A_1 A_2} \left[ U \cdot (\Psi^{ABR})^{\otimes n} \right],
\enq
and the last equality is true since $(\widetilde{U})^\dag \widetilde{U} = \eye^{A_2 B^n}$.
Lastly, we use the triangle inequality, \eqref{yae46}, and \eqref{yae45} to claim that
\begin{multline}
\left\| \widetilde{U}^{B_2 B^n \to B_1 \widetilde{B}_3^n B_3^n} \cdot \left\{ \cI^{A_2 \to B_2} \circ
\cC_W \left[ V \cdot (\Psi^{ABR})^{\otimes n} \right] \right\} - \Phi^{A_1 B_1} \otimes
(\Psi^{\widetilde{B}_3^n B_3^n R})^{\otimes n} \right\|_1 \\
\leqslant \Xi( \varepsilon_n ) + \Xi( \vartheta_n ).
\end{multline}
It follows that the protocol consists of Alice applying $\cC_W^{A \to A_1 A_2} \circ V^{A^n}$,
and Bob applies $\widetilde{U}$, albeit on $B_2 B^n$ instead of $A_2 B^n$.

It is now clear that if,
for $\alpha \in (1,2]$, $\delta_1, \delta_2 > 0$,
\begin{align}
\frac{\log|A_1|}{n} + \frac{\log |A_2|}{n} & = H_{\widetilde{\alpha}}(A)_\Psi +
|B| |R| \frac{ \log(n+1) }{n} + \delta_1, \\
\frac{\log|A_1|}{n} - \frac{\log |A_2|}{n} & = H_{\alpha}(A | R)_{\Psi} - |R| \frac{ \log(n+1) }{n} - \delta_2,
\end{align}
then the error decays exponentially in $n$ to zero.

The claim of the theorem now follows and we exhaust the entire achievable rate region as stipulated
by Theorem \ref{yatheorem1}.

Note that in view of the trivial protocol where one qubit transmitted across a noiseless qubit channel
from Alice and Bob generates one EPR pair shared by Alice and Bob (the reverse implication
is not true), it makes sense to keep the
quantum communication as small as possible, which is accomplished by making
$\alpha$ close to $1$, and $\delta_1, \delta_2$ close to $0$.
\end{proof}

\section{Fully quantum reverse Shannon (FQRS)}

The following definition is from Ref. \cite{fqsw-2009}.

\begin{definition}
A $(\Psi, \err, n)$ FQRS protocol consists of
$n$ copies of a pure state $\ket{\Psi}^{AA^\prime}$ (both $A$ and $A^\prime$ held by Alice),
a MES $\Phi^{A_1 B_1}$ shared between Alice ($A_1$) and Bob ($B_1$), a cptp map
$\cN^{A^\prime \to B}$ with Stinespring representation $V_{\cN}^{A^\prime \to BE}$ and
$\ket{\Psi}^{ABE} = V_{\cN}^{A^\prime \to BE} \ket{\Psi}^{AA^\prime}$, Alice applying an encoding
cptp map $\cE: A^{\prime n} A_1 \to A_2 E^n$, quantum communication across a noiseless quantum
channel from Alice to Bob $\cI^{A_2 \to B_2}$, and Bob applying a decoding cptp map
$\cD: B_1 B_2 \to B^n$ such that for
\beq
\rho^{A^n B^n E^n} \equiv \cD^{B_1 B_2 \to B^n}
\circ \cI^{A_2 \to B_2} \circ \cE^{A^{\prime n} A_1 \to A_2 E^n} [(\Psi^{AA^\prime})^{\otimes n}
\otimes \Phi^{A_1 B_1}],
\enq
\beq
\left\| \rho^{A^n B^n E^n} - (\Psi^{ABE})^{\otimes n} \right\|_1 \leqslant \err.
\enq
The number $(\log |B_2|)/n$ is called the {\bf quantum communication rate}
and $(\log |B_1|)/n$ is called the {\bf entanglement consumption rate} of the protocol.

A pair of real numbers $(\cR_Q,\cR_E)$ is called an {\bf achievable rate pair} if there
exist, for $n \to \infty$, FQRS protocols with quantum communication rate approaching
$\cR_Q$, entanglement consumption rate approaching $\cR_E$, and $\err$ approaching $0$.
\end{definition}

The achievable rates are described by the following theorem.

\begin{theorem}[Abeyesinghe \emph{et al}, 2009 \cite{fqsw-2009}]
\label{yatheorem2}
The following rates are achievable for the FQRS:
\beq
\cR_Q > \frac{1}{2} I(A:B)_\Psi \hspace{0.5in} \text{and} \hspace{0.5in}
\cR_E < \cR_Q + H(B|A)_\Psi.
\enq
\end{theorem}

We now provide the random coding exponents for the achievability of this protocol.

\begin{theorem}
For any $n \in \mathbb{N}$, there exists a $(\Psi, \err, n)$ FQRS protocol for any $\alpha \in (1,2]$,
$\delta_1, \delta_2 > 0$, such that
\begin{align}
\frac{\log |B_2|}{n} & =
\frac{1}{2} \left[ H_{\widetilde{\alpha}}(B)_\Psi - H_{\alpha}(B | A)_{\Psi} \right] + (|E|+1) |A| \frac{ \log(n+1) }{n} +
\frac{\delta_1 + \delta_2}{2}, \\
\frac{\log|B_1|}{n}  & = \frac{\log |B_2|}{n} + H_{\alpha}(B | A)_{\Psi} - |A| \frac{ \log(n+1) }{n} - \delta_2,
\end{align}
and the error approaches $0$ exponentially in $n$.
\end{theorem}
\begin{proof}
We note the insightful observation in Refs. \cite{devetak-prl-2006, fqsw-2009} that FQRS can be implemented
using ideas from FQSW.

Let $W^{B^n \to B_1 B_2}$, $|B_1| |B_2| \leqslant |B|^n$, be a full-rank partial isometry.
Then, using Corollary \ref{corollary1}, we claim that there exists a Unitary $U$ over $B^n$ such that
for $\alpha \in (1,2]$,
\begin{multline}
\label{yae13}
\left\| \tr_{B_1 B_2} \circ \mapone^{B^n \to B_1 B_2} \left[ U \cdot (\Psi^{ABE})^{\otimes n} \right] -
(\Psi^{AE})^{\otimes n} \right\|_1 \\
\leqslant 8 \exp \left\{ \frac{\alpha-1}{2\alpha} \Big[ |A| |E| \log(n+1) - n H_{\alpha}(B | AE)_{\Psi}
+ \Theta( \tr_{B_1 B_2} \circ \mapone ) \Big] \right\} \\
= 8 \exp \left\{ \frac{\alpha-1}{2\alpha} \Big[ |A| |E| \log(n+1) + n H_{\widetilde{\alpha}}(B)_{\Psi}
-  \log|B_1| |B_2| \Big] \right\} \equiv \varepsilon_n,
\end{multline}
and
\begin{multline}
\label{yae22}
\left\| \tr_{B_2} \circ \mapone^{B^n \to B_1 B_2} \left[ U \cdot (\Psi^{AB})^{\otimes n} \right] - \pi^{B_1} \otimes
(\Psi^A)^{\otimes n} \right\|_1 \\
\leqslant 8 \exp \left\{ \frac{\alpha-1}{2\alpha} \left[ |A| \log(n+1) - n H_{\alpha}(B | A)_{\Psi}
+ \log \frac{|B_1|}{|B_2|} \right] \right\} \equiv \vartheta_n.
\end{multline}
Using \eqref{yae22} and Lemma \ref{yal11}, we claim that there exists an isometry
$\widetilde{U}^{B_2 E^n \to A_1 \widetilde{B}^n \widetilde{E}^n}$ such that
\beq
\label{yae56}
\left\| \widetilde{U} \cdot
\mapone^{B^n \to B_1 B_2} \left[ U \cdot (\Psi^{ABE})^{\otimes n} \right] - \Phi^{A_1 B_1} \otimes
(\Psi^{A \widetilde{B} \widetilde{E}} )^{\otimes n} \right\|_1 \leqslant \Xi( \vartheta_n ).
\enq
Using the compressive map $\cC_{\widetilde{U}^\dag} : A_1 \widetilde{B}^n \widetilde{E}^n \to
A_2 E^n$, \eqref{yae56}, and monotonicity, we get
\beq
\label{yae57}
\left\| W^\dag \cdot \mapone^{B^n \to B_1 A_2} \left[ U \cdot (\Psi^{AB})^{\otimes n} \right] -
W^\dag \cdot \cC_{\widetilde{U}^\dag} \left[ \Phi^{A_1 B_1} \otimes
(\Psi^{A \widetilde{B} \widetilde{E}} )^{\otimes n} \right] \right\|_1 \leqslant \Xi( \vartheta_n ).
\enq
Using \eqref{yae13} and Lemma \ref{yal11}, we claim that there exists a Unitary $V$ over $B^n$ such that
\beq
\label{yae58}
\left\| W^\dag \cdot \mapone^{B^n \to B_1 B_2} \left[ U \cdot (\Psi^{ABE})^{\otimes n} \right] -
V \cdot (\Psi^{ABE})^{\otimes n} \right\|_1 \leqslant \Xi( \varepsilon_n ).
\enq
Using the triangle inequality, \eqref{yae57}, and \eqref{yae58}, we now have
\beq
\left\| (V^\dag W^\dag) \circ \cI^{A_2 \to B_2} \circ \cC_{\widetilde{U}^\dag} \left[ \Phi^{A_1 B_1} \otimes
(\Psi^{A \widetilde{B} \widetilde{E}} )^{\otimes n} \right] - (\Psi^{ABE})^{\otimes n} \right\|_1
\leqslant \Xi( \varepsilon_n ) + \Xi( \vartheta_n ).
\enq
Hence, the FQRS protocol consists of Alice applying
\beq
\cE^{A^{\prime n} A_1 \to A_2 E^n} =
\cC_{\widetilde{U}^\dag}^{A_1 \widetilde{B}^n \widetilde{E}^n \to A_2 E^n} \circ
\left( V_\cN^{A^\prime \to \widetilde{B} \widetilde{E}} \right)^{\otimes n},
\enq
and Bob applying $\cD^{B_1 B_2 \to B^n} = V^\dag W^\dag$. The claim now follows readily.
\end{proof}

\section{Quantum state merging (QSM)}

\begin{definition}
A $(\Psi, \err, n)$ QSM protocol consists of
$n$ copies of a pure state $\ket{\Psi}^{ABR}$ shared between Alice ($A$),
Bob ($B$), and reference ($R$) inaccessible to both Alice and Bob,
a MES $\Phi^{A_0 B_0}$ shared between Alice ($A_0$)
and Bob ($B_0$), and a local operation and classical communication
(locc) quantum operation
$\cM: A^n A_0 \otimes B^n B_0 \to A_1 \otimes B_1 \widetilde{B}_2^n B_2^n$ such that for
\beq
\rho^{A_1 B_1 \widetilde{B}_2^n B_2^n R^n} \equiv
\cM \left[ (\Psi^{ABR})^{\otimes n} \otimes \Phi^{A_0 B_0} \right],
\enq
\beq
\left\| \rho^{A_1 B_1 \widetilde{B}_2^n B_2^n R^n} - \Phi^{A_1 B_1} \otimes
\Psi^{\widetilde{B}_2^n B_2^n R^n} \right\| \leqslant \err,
\enq
where $\Phi^{A_1 B_1}$ is a MES shared between Alice ($A_1$) and Bob ($B_1$).
The number
$(\log |A_0| - \log |A_1|)/n$ is called the {\bf entanglement rate} of the protocol.
A real number $\cR_E$ is called an {\bf achievable rate} if there exist, for
$n \to \infty$, QSM protocols of rate approaching $\cR_E$
and $\err$ approaching $0$.
\end{definition}

The achievable rate is given in the next theorem.

\begin{theorem}[Horodecki \emph{et al}, 2005 \cite{state-merging-2005}]
The following rates are achievable:
\beq
\cR_E > H(A|B)_\Psi.
\enq
Furthermore, there
exists a QSM protocol that achieves this merging cost using one-way locc
with a classical communication cost of $I(A:R)_\Psi$ per input copy.
\end{theorem}

We prove the following.

\begin{theorem}
For any $n \in \mathbb{N}$, there exists a $(\Psi, \err, n)$ QSM
protocol using one-way locc
for arbitrary $\delta_1, \delta_2 > 0$, $\alpha \in (1,2]$, such that
\begin{align}
\frac{\log |A_0| - \log |A_1|}{n} & = H_{\widetilde{\alpha}}(A | B)_{\Psi} + |R| \frac{\log (n+1)}{n} + \delta_1,
\end{align}
and a classical communication cost of at most
\begin{align}
H_{\widetilde{\alpha}}(A)_\Psi - H_{\alpha}(A | R)_{\Psi} + \frac{(|B|+1)|R| \log (n+1) + 2}{n} +
\delta_1 + \delta_2,
\end{align}
with the error approaching $0$ exponentially in $n$.
\end{theorem}
\begin{proof}
Our line of attack is similar to that in Ref. \cite{state-merging-2005},
Corollary 3.11 in Ref. \cite{dupuis-thesis}, and Theorem 5.2 in Ref. \cite{dupuis-et-al-2010}.

Let $W^{A^n \to E}$, $|E| \leqslant |A|^n$, be a full-rank partial isometry.
Let $\zeta \equiv \frac{|E| |A_0|}{|A_1|}$, $J \equiv \lceil \zeta \rceil$,
and let $M_x^{E  A_0 \to A_1}$, $x=1,...,J$, $|A_1| \leqslant |A_0| |E|$,
be a set of measurement operators such that $\sum_{x=1}^{J} M_x^\dagger M_x = \eye^{EA_0}$,
where each $M_x$ (except possibly when $x = J$) is a full-rank partial isometry.

For any orthonormal basis $\{\ket{x}^{X}\}$, $x=1,...,J$, we define
\begin{align}
\cE^{E A_0 \to X A_1}(\sigma^{E A_0}) & \equiv \sum_{x=1}^{J} \ketbra{x}^{X}
\otimes (M_x \cdot \sigma^{E A_0})  \\
\omega^{X A_1} &
\equiv \cE^{E A_0 \to X A_1} \circ \mapone^{A^n \to E} \left( \pi^{A^n A_0} \right).
\end{align}
We have
\begin{align}
\omega^{X A_1} & = \cE^{E A_0 \to X A_1} \left( \pi^{E A_0} \right) \\
& = \frac{ 1 }{\zeta} \sum_{x=1}^{J} \ketbra{x}^X \otimes \pi^{A_1}
- \ketbra{J}^X \otimes \frac{P^{A_1}}{|E| |A_0|} \\
& = \frac{J}{\zeta} \pi^{X A_1} - \ketbra{J}^X \otimes \frac{P^{A_1}}{|E| |A_0|},
\end{align}
where $P^{A_1}$ is a projector with rank $< A_1$. Note that
\beq
\label{yae61}
\left\| \omega^{X A_1} - \pi^{X A_1} \right\|_1
\leqslant \left\| \left( \frac{J}{\zeta} - 1 \right) \pi^{X A_1} \right\|_1 +
\left\| \frac{P^{A_1}}{|E| |A_0|} \right\|_1
< \left( \frac{J}{\zeta} -1 \right) + \frac{ 1 }{\zeta}
< \frac{ 2 }{\zeta},
\enq
where the first inequality follows from the triangle inequality, the second one from
$\tr P^{A_1} < |A_1|$, and the third one from $J - \zeta < 1$.

Invoking Corollary \ref{corollary1}, we first claim that there exists a Unitary $U^{A^n A_0}$ such that for
any $\alpha \in (1,2]$,
\begin{multline}
\label{yae60}
\left\| \cE^{A^n A_0 \to A_1 X} \circ \mapone^{A^n \to E}
\left\{ U^{A^n A_0} \cdot \left[ (\Psi^{AR})^{\otimes n} \otimes \pi^{A_0})
\right] \right\} - \omega^{A_1 X} \otimes (\Psi^R)^{\otimes n} \right\|_1 \\
\leqslant  8 \exp\left\{ \frac{\alpha-1}{2\alpha} \Big[ |R| \log (n+1) -
H_{\alpha}(A^n A_0 | R^n)_{\Psi^{\otimes n}} +
\Theta( \cE \circ \mapone) \Big] \right\} \hspace{1.25in} \\
\leqslant  8 \exp\left\{ \frac{\alpha-1}{2\alpha} \Big[ |R| \log (n+1) + n H_{\widetilde{\alpha}}(A | B)_{\Psi} -
(\log |A_0| - \log |A_1|) \Big] \right\}
\equiv \vartheta_n,
\end{multline}
(where in the second inequality, we have used
$- H_{\alpha}(A^n A_0 | R^n)_{\Psi^{\otimes n}}
= - n H_{\alpha}(A|R)_\Psi - \log |A_0| = n H_{\widetilde{\alpha}}(A|B)_\Psi - \log |A_0|$,
and, from Lemma \ref{yal8}, $\Theta( \cE \circ \mapone) \leqslant$ $\log|A_1|$) and
\begin{multline}
\label{yae47}
\left\| \tr_{EA_0} \circ \mapone^{A^n \to E} \left\{ U^{A^n A_0}
\cdot \left[ (\Psi^{A BR})^{\otimes n} \otimes \Phi^{A_0 B_0} \right]
\right\} - (\Psi^{BR})^{\otimes n} \otimes \pi^{B_0} \right\|_1 \\
\leqslant 8 \exp \left\{ \frac{\alpha - 1}{2\alpha} \Big[ |B| |R| \log (n+1)
- H_\alpha(A^n A_0 |B^n R^n B_0)_{\Psi^{\otimes n} \otimes \Phi} - \log (|A_0||E|) \Big] \right\} \\
\leqslant 8 \exp \left\{ \frac{\alpha - 1}{2\alpha} \Big[ |B| |R| \log (n+1)
+ n H_{\widetilde{\alpha}}(A)_\Psi - \log|E| \Big] \right\} \equiv \varepsilon_n,
\end{multline}
where in the first inequality, we have used $\nu_{(\Psi^{BR})^{\otimes n} \otimes \pi^{B_0}}
= \nu_{(\Psi^{BR})^{\otimes n}}$, and the second inequality follows from
\begin{multline}
- H_\alpha(A^n A_0 |B^n R^nB_0)_{\Psi^{\otimes n} \otimes \Phi} - \log (|A_0||E|) \\
\leqslant - H_\alpha(A^n|B^n R^n)_{\Psi^{\otimes n}} - H_\alpha(A_0|B_0)_\Phi - \log (|A_0||E|)
= n H_{\widetilde{\alpha}}(A)_\Psi - \log|E|.
\end{multline}
\eqref{yae47} implies using Lemma \ref{yal11} that there exists
a Unitary $V^{A^n A_0 \to A^n A_0}$ such that
\begin{multline}
\label{yae59}
\left\| W^\dag \cdot \mapone^{A^n \to E} \left\{ U^{A^n A_0}
\cdot \left[ (\Psi^{A BR})^{\otimes n} \otimes \Phi^{A_0 B_0} \right]
\right\} - V^{A^n A_0 \to A^n A_0} \cdot (\Psi^{ABR})^{\otimes n} \otimes \Phi^{A_0 B_0} \right\|_1 \\
\leqslant \Xi( \varepsilon_n ),
\end{multline}
and, using monotonicity, this implies that
\beq
\label{yae62}
\left\| \mapone^{A^n \to E} \left\{ U^{A^n A_0}
\cdot \left[ (\Psi^{A BR})^{\otimes n} \otimes \Phi^{A_0 B_0} \right]
\right\} - \cC_W \left[ V \cdot (\Psi^{ABR})^{\otimes n} \otimes \Phi^{A_0 B_0} \right]
\right\|_1
\leqslant \Xi( \varepsilon_n ).
\enq
We now have
\begin{multline}
\label{yae10}
\left\| \cE^{E A_0 \to A_1 X} \circ \mapone^{A^n \to E}
\left[ U \cdot (\Psi^{AR})^{\otimes n} \otimes \pi^{A_0} \right]
- \pi^{A_1 X} \otimes (\Psi^R)^{\otimes n} \right\|_1 \\
\leqslant 
\left\| \cE^{E A_0 \to A_1 X} \circ \mapone^{A^n \to E}
\left[ U \cdot (\Psi^{AR})^{\otimes n} \otimes \pi^{A_0} \right] -
\omega^{A_1 X} \otimes (\Psi^R)^{\otimes n} \right\|_1 \\
+ \left\| \omega^{A_1 X} \otimes (\Psi^R)^{\otimes n} - \pi^{A_1 X} \otimes (\Psi^R)^{\otimes n}
\right\|_1 \\
\leqslant \vartheta_n + \frac{ 2 }{\zeta} \equiv \beta_n, \hspace{3.96in}
\end{multline}
where the first inequality follows from the triangle inequality, and in the second
inequality, the first term is upper bounded using \eqref{yae60}, and the last term by using
\eqref{yae61}. Let
\begin{align}
\xi^{A_1 B^n B_0 R^n}_x & \equiv \frac{J |A|^n}{|E|} (M_x W U) \cdot
(\Psi^{ABR})^{\otimes n} \otimes \Phi^{A_0 B_0} \\
\sigma^{XA_1 B^n B_0 R^n} & \equiv \cE^{E A_0 \to A_1 X} \circ \mapone^{A^n \to E}
\left[ U \cdot (\Psi^{ABR})^{\otimes n} \otimes \Phi^{A_0 B_0} \right] \nonumber \\
& = \sum_{x=1}^{J} \frac{1}{J} \ketbra{x}^X \otimes \xi^{A_1 B^n B_0 R^n}_x.
\end{align}
Note that $\xi^{A_1 B^n B_0 R^n}_x$ is a pure state for all $x$.
Let $\varepsilon^\prime_x \equiv \left\| \xi_x^{A_1 R^n} - \pi^{A_1} \otimes (\Psi^R)^{\otimes n} \right\|_1$.
We now have
\beq
\label{yae11}
\beta_n \geqslant \left\| \sigma^{X A_1 R^n} - \pi^{A_1 X} \otimes (\Psi^R)^{\otimes n} \right\|_1
= \sum_{x=1}^{J} \frac{\varepsilon^\prime_x}{J}.
\enq
From Lemma \ref{yal11}, let $V_x^{B^n B_0 \to B_1 \widetilde{B}_2^n B_2^n}$ be an isometry
such that
\beq
\label{yae64}
\left\| V_x^{B^n B_0 \to B_1 \widetilde{B}_2^n B_2^n} \cdot \xi_x^{A_1 B^n B_0 R^n} - \Phi^{A_1 B_1}
\otimes (\Psi^{\widetilde{B}_2 B_2 R})^{\otimes n} \right\|_1 \leqslant \Xi(\varepsilon^\prime_x).
\enq
Define a cptp map
\beq
\cD^{X B^n B_0 \to B_1 \widetilde{B}_2^n B_2^n} \left( \sum_{x=1}^{J}
\ketbra{x}^X \otimes \Upsilon_x^{B^n B_0}
\right) \equiv \sum_{x=1}^{J} (V_x \cdot \Upsilon_x^{B^n B_0}).
\enq
We now have
\begin{align}
%\label{yae64}
\Big\| \cD \circ \cE \circ \cT_W & \left[ U \cdot (\Psi^{ABR})^{\otimes n} \otimes \Phi^{A_0 B_0} \right]
- \Phi^{A_1 B_1} \otimes (\Psi^{\widetilde{B}_2 B_2 R})^{\otimes n} \Big\|_1 \\
& = \left\| \sum_{x=1}^{J} \frac{1}{J}
 V_x^{B^n B_0 \to B_1 \widetilde{B}_2^n B_2^n} \cdot \xi_x^{A_1 B^n B_0 R^n} - \Phi^{A_1 B_1}
\otimes (\Psi^{\widetilde{B}_2 B_2 R})^{\otimes n} \right\|_1 \\
& \leqslant \sum_{x=1}^{J} \frac{1}{J}
\left\| V_x^{B^n B_0 \to B_1 \widetilde{B}_2^n B_2^n} \cdot \xi_x^{A_1 B^n B_0 R^n} - \Phi^{A_1 B_1}
\otimes (\Psi^{\widetilde{B}_2 B_2 R})^{\otimes n} \right\|_1 \\
& \leqslant \sum_{x=1}^{J} \frac{1}{J} \Xi(\varepsilon^\prime_x) \\
& \leqslant \sum_{x=1}^{J} \frac{1}{J} \left[
2 \sqrt{\varepsilon^\prime_x} + \sqrt{2} (\varepsilon^\prime_x)^{3/4}
+ \varepsilon^\prime_x \right] \\
& \leqslant 2 \sqrt{ \sum_{x=1}^{J} \frac{\varepsilon^\prime_x}{J} } +
\sqrt{2} \left( \sum_{x=1}^{J} \frac{\varepsilon^\prime_x}{J} \right)^{3/4} +
\sum_{x=1}^{J} \frac{\varepsilon^\prime_x}{J} \\
& \leqslant 2 \sqrt{\beta_n} + \sqrt{2} (\beta_n)^{3/4} + \beta_n,
\end{align}
where the first inequality follows from the convexity of the trace norm, the second inequality
follows from \eqref{yae64}, the third inequality follows since
\beq
\Xi(\varepsilon) = \sqrt{\varepsilon (2 + \varepsilon + 2 \sqrt{1+ \varepsilon})}
\leqslant 2 \sqrt{\varepsilon} + \sqrt{2} (\varepsilon)^{3/4} + \varepsilon,
\enq
and the fourth inequality from the concavity of $x \mapsto x^y$, $y \in [0,1]$,
and the last inequality follows from \eqref{yae11}.
Using \eqref{yae62} and the triangle inequality, we have
\begin{multline}
\left\| \cD \circ \cE \circ \cC_W \left[ V \cdot (\Psi^{ABR})^{\otimes n} \otimes \Phi^{A_0 B_0} \right]
- \Phi^{A_1 B_1} \otimes (\Psi^{\widetilde{B}_2 B_2 R})^{\otimes n} \right\|_1 \\
\leqslant \Xi( \varepsilon_n ) + 2 \sqrt{\beta_n} + \sqrt{2} (\beta_n)^{3/4} + 
\beta_n.
\end{multline}
Alice performs $\cE^{A^n A_0 \to A_1 X} \circ \cC_W \circ V$ and
Bob performs $\cD^{X B_0 B^n \to B_1 \widetilde{B}_2^n B_2^n}$.
Note that $J = \lceil \frac{|E| |A_0|}{|A_1|} \rceil$ $\leqslant \max\{1,
\frac{2 |E| |A_0|}{|A_1|}\}$, determines the classical communication cost.
We have now shown the existence of a state merging protocol using one-way locc
for arbitrary $\delta_1, \delta_2 > 0$, $\alpha \in (1,2]$, with
\begin{align}
\frac{1}{n} \log \frac{|A_0|}{|A_1|} & = H_{\widetilde{\alpha}}(A | B)_{\Psi} + |R| \frac{\log (n+1)}{n} + \delta_1 \\
\frac{\log |E|}{n} & = H_{\widetilde{\alpha}}(A)_\Psi + |B| |R| \frac{\log (n+1)}{n} + \delta_2 \\
\frac{\log J}{n} &
\leqslant  H_{\widetilde{\alpha}}(A)_\Psi - H_{\alpha}(A | R)_{\Psi} + \frac{(|B|+1)|R| \log (n+1) + 1}{n}
+ \delta_1 + \delta_2
\end{align}
that has the error converging to $0$ exponentially in $n$.
\end{proof}

\section{Entanglement-assisted quantum communication with side information at the
transmitter (Father with side information at the transmitter)}

The definitions are directly from Ref. \cite{dupuis-thesis}.
\begin{definition}
Let $\cN^{A^\prime S \to B}$ be a cptp map with Stinespring dilation $V_\cN^{A^\prime S \to BE}$
and $\ket{\Upsilon}^{S S^\prime}$ be a pure state. Then the transmitter encodes its information
contained in $\rho^{A_1 R} \in \densitymatrix(\cH_{A_1R})$ using a cptp map
$\cE^{A_1 S^\prime \to A^\prime}$, and the output of the channel is $\rho^{BR} =$
$\cN^{A^\prime S \to B} \circ \cE^{A_1 S^\prime \to A^\prime} (\rho^{A_1R} \otimes \Upsilon^{S S^\prime})$.
We denote this channel by $\{\cN^{A^\prime S \to B}, \ket{\Upsilon}^{S S^\prime}\}$.
\end{definition}

\begin{definition}
\label{def1}
A $(\{\cN, \ket{\Upsilon}\},\err,n)$ father protocol with side information at the transmitter consists of $n$ copies of two MES
$\Phi^{A_0 R}$ and $\Phi^{A_1 B_1}$, where Alice has $A_0, A_1$, Bob has $B_1$,
and the reference $R$ is inaccessible to both Alice and Bob, Alice applying an encoding
map $\cE^{A_0^n A_1^n S^{\prime n} \to A^{\prime n}}$ to
$(\Phi^{A_0 R} \otimes \Phi^{A_1 B_1} \otimes \Upsilon^{S S^\prime})^{\otimes n}$,
$n$ uses of the channel with side information at the transmitter
$\{\cN^{A^\prime S \to B}, \ket{\Upsilon}^{S S^\prime}\}$, and
Bob applying a decoding map $\cD^{B^n B_1^n \to B_2^n}$ such that for
\beq
\rho^{B_2^n R^n} \equiv \cD^{B^n B_1^n \to B_2^n} \circ ( \cN^{A^\prime S \to B} )^{\otimes n} \circ
\cE^{A_0^n A_1^n S^{\prime n} \to A^{\prime n}} (\Phi^{A_0 R} \otimes \Phi^{A_1 B_1} \otimes
\Upsilon^{S S^\prime})^{\otimes n},
\enq
\beq
\left\| \rho^{B_2^n R^n} - (\Phi^{B_2 R})^{\otimes n} \right\|_1 \leqslant \err.
\enq
The number $\log |B_1|$ is called the {\bf entanglement rate} of the protocol
and $\log |R|$ is called the {\bf quantum communication rate} of the protocol.

A pair of real numbers $(\cR_Q,\cR_E)$ is called an {\bf achievable rate pair} if there
exist, for $n \to \infty$, protocols with quantum communication rate approaching
$\cR_Q$, entanglement gain rate approaching $\cR_E$, and $\err$ approaching $0$.
\end{definition}

The achievable rates are described by the following theorem.

\begin{theorem}[Dupuis, 2009 \cite{dupuis-thesis}]
\label{yatheorem7}
Let $\ket{\Psi}^{CA A^{\prime}S}$ be a pure state with $\cH_{A} = \cH_R \otimes \cH_{B_1}$ such that
$\Psi^S = \Upsilon^S$, and $\ket{\Psi}^{CA BE}$ $= V_\cN^{A^\prime S \to BE} \ket{\Psi}^{CA A^{\prime}S}$.
The following rates are achievable:
\begin{align}
\cR_Q + \cR_E & < H(A|S)_\Psi \\
\cR_Q - \cR_E  & < -H(A|B)_\Psi.
\end{align}
\end{theorem}

We now have the following theorem.
\begin{theorem}
\label{yatheorem4}
For any $n \in \mathbb{N}$, and $\Psi$ as defined in Theorem \ref{yatheorem7},
there exists a $(\{\cN, \ket{\Upsilon}\}, \err, n)$ Father protocol with side information at the transmitter such that for any
$\alpha \in (1,2]$ and $\delta_1, \delta_2 > 0$,
\begin{align}
\log|R| + \log|B_1| & = H_\alpha(A|S)_\Psi - |S| \frac{\log(n+1)}{n} - \delta_1 \\
\log|R| - \log|B_1|  & = -H_{\widetilde{\alpha}}(A | B)_{\Psi} - |C| |E| \frac{\log(n+1)}{n} - \delta_2,
\end{align}
and the error approaches $0$ exponentially in $n$.
\end{theorem}
\begin{proof}
We first claim using Corollary \ref{corollary1} that there exists a
Unitary $U$ on $R^nB_1^n$ such that
\begin{multline}
\label{yae27}
\left\| \tr_{B_1^n} \left[ U \cdot (\Psi^{C R B_1 E})^{\otimes n} \right] -
\left( \pi^R \otimes \Psi^{CE} \right)^{\otimes n} \right\|_1 \\
\leqslant 8 \exp\left\{ \frac{\alpha-1}{2\alpha} \left[ |C||E| \log(n+1) - n H_{\alpha}(A | CE)_{\Psi}
+ n \log \frac{|R|}{|B_1|} \right] \right\} \\
= 8 \exp\left\{ \frac{\alpha-1}{2\alpha} \left[ |C||E| \log(n+1) + n H_{\widetilde{\alpha}}(A | B)_{\Psi}
+ n \log \frac{|R|}{|B_1|} \right] \right\} \equiv \varepsilon_n,
\end{multline}
and
\begin{multline}
\label{yae28}
\left\| U \cdot (\Psi^{R B_1 S})^{\otimes n} - (\pi^{R B_1})^{\otimes n} \otimes (\Upsilon^S)^{\otimes n}
\right\|_1 =
\left\| U \cdot (\Psi^{R B_1 S})^{\otimes n} - (\pi^{R B_1})^{\otimes n} \otimes (\Psi^S)^{\otimes n}
\right\|_1 \\
\leqslant 8 \exp\left\{ \frac{\alpha-1}{2\alpha} \Big[ |S| \log(n+1) - n H_{\alpha}(A|S)_{\Psi}
+ n \log (|R||B_1|) \Big] \right\} \equiv \vartheta_n,
\end{multline}
where in \eqref{yae28}, we have used $\Psi^S = \Upsilon^S$, and it
follows from \eqref{yae28} and Lemma \ref{yal11} that there exists an isometry
$V_1^{A_0^n A_1^n S^{\prime n} \to A^{\prime n} C^n}$ such that
\beq
\label{yae29}
\left\| U \cdot (\Psi^{C R B_1 A^\prime S})^{\otimes n} -
V_1^{A_0^n A_1^n S^{\prime n} \to A^{\prime n} C^n} \cdot
(\Phi^{A_0 R} \otimes \Phi^{A_1 B_1} \otimes \Upsilon^{S^\prime S})^{\otimes n} \right\|_1
\leqslant 2 \sqrt{\vartheta_n}.
\enq
Using the triangle inequality, \eqref{yae27}, \eqref{yae29},
and monotonicity, we have
\begin{multline}
\left\| \tr_{B_1^n B^n} \left\{ \left[ (V_\cN)^{\otimes n} V_1^{A_0^n A_1^n S^{\prime n} \to A^{\prime n} C^n}
\right] \cdot (\Phi^{A_0 R} \otimes \Phi^{A_1 B_1} \otimes \Upsilon^{S^\prime S})^{\otimes n} \right\} -
\left( \pi^R \otimes \Psi^{CE} \right)^{\otimes n} \right\|_1 \\
\leqslant \varepsilon_n + 2 \sqrt{\vartheta_n}.
\end{multline}
Hence there exists an isometry $V_2^{B_1^n B^n \to B_2^n \widetilde{A}^n \widetilde{B}^n}$ such that
for some purifications $\Phi^{R B_2}$ and $\Psi^{\widetilde{A} \widetilde{B} C E}$ of
$\pi^R$ and $\Psi^{CE}$ respectively, we have
\begin{multline}
\left\| \left[ V_2^{B_1^n B^n \to B_2^n \widetilde{A}^n \widetilde{B}^n}
(V_\cN)^{\otimes n} V_1^{A_0 A_1 \to A^{\prime n}} \right] \cdot
(\Phi^{A_0 R} \otimes \Phi^{A_1 B_1} \otimes \Upsilon^{S^\prime S})^{\otimes n}
- (\Phi^{R B_2} \otimes \Psi^{\widetilde{A} \widetilde{B} C E})^{\otimes n} \right\|_1 \\
\leqslant 2 \sqrt{\varepsilon_n + 2 \sqrt{\vartheta_n}}
\end{multline}
and hence,
\begin{multline}
\left\| \tr_{\widetilde{A}^n \widetilde{B}^n C^n} \left\{ V_2 
\cdot (\cN)^{\otimes n} \left[ V_1 \cdot
(\Phi^{A_0 R} \otimes \Phi^{A_1 B_1} \otimes \Upsilon^{S^\prime S})^{\otimes n} \right] \right\} - (\Phi^{R B_2})^{\otimes n}
\right\|_1 \\
= \left\| \tr_{\widetilde{A}^n \widetilde{B}^n C^n E^n} \left\{ \left[ V_2 (V_\cN)^{\otimes n} V_1 \right] \cdot
(\Phi^{A_0 R} \otimes \Phi^{A_1 B_1} \otimes \Upsilon^{S^\prime S})^{\otimes n} \right\} -
(\Phi^{R B_2})^{\otimes n} \right\|_1 \\
\leqslant 2 \sqrt{\varepsilon_n + 2 \sqrt{\vartheta_n}}. \hspace{4.1in}
\end{multline}
It is now clear that Alice just applies $\tr_{C^n} \circ V_1^{A_0^n A_1^n S^{\prime n} \to A^{\prime n} C^n}$ and
Bob applies
$\tr_{\widetilde{A}^n \widetilde{B}^n} \circ V_2^{B_1^n B^n \to B_2^n \widetilde{A}^n \widetilde{B}^n}$.
The claim now follows readily.
\end{proof}

We now have the following corollary to obtain the regularized expressions by additional blocking.
\begin{corollary}
For any $m, n \in \mathbb{N}$, a pure state $\Psi^{C AA^{\prime m}S^m}$ with
$\cH_{A} = \cH_R \otimes \cH_{B_1}$ such that $\Psi^{S^m} = (\Upsilon^S)^{\otimes m}$
and $\ket{\Psi}^{CA B^mE^m}$ $= (V_\cN^{A^\prime S \to BE})^{\otimes m}
\ket{\Psi}^{CA A^{\prime m}S^m}$,
there exists a $(\{\cN, \ket{\Upsilon}\}, \err, mn)$ Father protocol with side information at the transmitter such that for any
$\alpha \in (1,2]$ and $\delta_1, \delta_2 > 0$,
\begin{align}
\label{yae68}
\frac{\log|R|}{m} + \frac{\log|B_1|}{m} & = \frac{H_\alpha(A|S^m)_\Psi}{m} - |S| \frac{\log(mn+1)}{mn} - \delta_1 \\
\label{yae69}
\frac{\log|R|}{m} - \frac{\log|B_1|}{m}  & = -\frac{H_{\widetilde{\alpha}}(A | B^m)_{\Psi}}{m} -
|C| |E|^m \frac{\log(n+1)}{mn} - \delta_2,
\end{align}
and the error approaches $0$ exponentially in $mn$.
\end{corollary}
We omit the proof. Rather than blindly applying Theorem \ref{yatheorem4},
we need to use \linebreak $\nu_{(\Upsilon^S)^{\otimes mn}} \leqslant (mn+1)^{|S|}$.
%Note that $|C| |E|^m$ factor appears above (unless $\Psi$ is a product state)
%and $n$ has to be large to make it smaller.
The number $m$ serves two purposes. Firstly, it enables a better approximation to
the optimal rates, and, secondly, it allows for finer approximation to the \renyi quantities through the choices of
$|R|$ and $|B_1|$.

Note that by choosing $|B_1|=1$, we get {\bf entanglement-unassisted quantum communication}
as a special case of the above and for any $\alpha \in (1,2]$ and $\delta_1, \delta_2 > 0$,
the rate is given by
\begin{multline}
\frac{\log|R|}{m} =  \min \Big\{ \frac{H_\alpha(A|S^m)_\Psi}{m} - |S| \frac{\log(mn+1)}{mn} - \delta_1, \\
-\frac{H_{\widetilde{\alpha}}(A | B^m)_{\Psi}}{m} -
|C| |E|^m \frac{\log(n+1)}{mn} - \delta_2 \Big\}.
\end{multline}

Assuming $H_\alpha(A|S^m)_\Psi \geqslant - H_{\widetilde{\alpha}}(A | B^m)_{\Psi}$,
the rate for {\bf quantum communication assisted by
unlimited entanglement} for any $\alpha \in (1,2]$ and
$\delta > 0$ is given by
\begin{multline}
\label{yae67}
\frac{\log|R|}{m} = \frac{H_\alpha(A|S^m)_\Psi - H_{\widetilde{\alpha}}(A | B^m)_{\Psi}}{2m}
- \frac{|S| \log(mn+1) + |C| |E|^m \log(n+1)}{2mn} - \delta.
\end{multline}

\begin{definition}
A $(\{\cN, \ket{\Upsilon}\},\err,n)$ {\bf entanglement-assisted classical communication} protocol with side information at
the transmitter consists of $n$ copies of an MES $\Phi^{A_2B_2}$, where Alice has $A_2$ and
Bob has $B_2$, Alice having a random variable $X$ uniformly distributed over a set $\cX$ that
models the information, Alice applying an encoding map $\cE_x^{A_2^n S^{\prime n} \to A^{\prime n}}$,
$x \in \cX$, if $X = x$, $n$ uses of the channel with side information at the transmitter
$(\cN^{A^\prime S \to B}, \ket{\Upsilon}^{S S^\prime})$, and
Bob applying a POVM (positive operator-valued measure) $\{ \Lambda^{B^n B_2^n}_{x^\prime}$,
$x^\prime \in \cX \}$, such that for
\beq
\Pr\{ x^\prime | x \} \equiv \tr \, \Lambda_{x^\prime}^{B^n B_2^n} \left[ ( \cN^{A^\prime S \to B} )^{\otimes n}
\circ \cE_x^{A_2^n \to A^{\prime n}} (\Phi^{A_2 B_2} \otimes \Upsilon^{S S^\prime})^{\otimes n} \right],
\enq
\beq
\frac{1}{|\cX|} \sum_x (1 - \Pr \{ x | x \} ) \leqslant \err.
\enq
The number $(\log |\cX|)/n$ is called the {\bf classical communication rate} of the protocol.

A real number $\cR_C$ is called an {\bf achievable rate} if there
exist, for $n \to \infty$, any choice of $|A_2|$, protocols with classical communication rate approaching
$\cR_C$ and $\err$ approaching $0$.
\end{definition}

Note that the capacity for this protocol was obtained in Ref.
\cite{dupuis-thesis}.
We now provide the random coding exponents for the entanglement-assisted classical communication.
\begin{corollary}
For any $m, n \in \mathbb{N}$, a pure state $\Psi^{C AA^{\prime m}S^m}$ with
$\cH_{A} = \cH_R \otimes \cH_{B_1}$ such that $\Psi^{S^m} = (\Upsilon^S)^{\otimes m}$
and $\ket{\Psi}^{CA B^mE^m}$ $= (V_\cN^{A^\prime S \to BE})^{\otimes m}
\ket{\Psi}^{CA A^{\prime m}S^m}$,
there exists a $(\{\cN, \ket{\Upsilon}\}, \err, mn)$ \linebreak {\bf entanglement-assisted classical communication} protocol with
side information at the transmitter such that for any $\alpha \in (1,2]$ and $\delta > 0$,
the rate per channel use is given by
\begin{align}
\frac{\log|\cX|}{mn} = \frac{H_\alpha(A|S^m)_\Psi}{m} -\frac{H_{\widetilde{\alpha}}(A | B^m)_{\Psi}}{m}
- |S| \frac{\log(mn+1)}{mn} - |C| |E|^m \frac{\log(n+1)}{mn} - \delta,
\end{align}
and the error approaches $0$ exponentially in $mn$.
\end{corollary}
\begin{proof}
We follow the well-understood strategy to encapsulate the entanglement-assisted quantum communication
protocol in the qudit superdense coding protocol. We follow the notation in Definition \ref{def1}.
Let $\cH_{A_2} = \cH_{A_0} \otimes \cH_{A_1}$ and
$\cH_{B_2} = \cH_{R} \otimes \cH_{B_1}$. Alice
has access to $A_0, A_1$ and Bob has access to $R, B_1$. Let
$V_i \in \bbU(R^n)$ such that $\tr V_i^\dag V_j = |R|^n \delta_{i,j}$. Alice chooses
$|\cX| = |R|^{2n}$, and, for $X = x$, Alice applies
$V_x$ over $R^n$ on $(\Phi^{A_0 R})^{\otimes n}$ (Alice does this by exploiting the Schmidt
symmetry) and passes that MES as input to the father protocol that uses the
channel $m \times n$ times.
At the end of the father protocol, we have a state $\rho_x^{B_1^{mn} R^{mn}}$ such that
\beq
\left\| \rho_x^{B_1^{mn} R^{mn}} - V_x \cdot (\Phi^{B_1 R})^{\otimes mn} \right\|_1 \leqslant \beta_{mn},
\enq
where $\beta_{mn} = 2 \sqrt{\varepsilon_{m,n} + 2 \sqrt{\vartheta_{m,n}}}$, and
\begin{align}
\varepsilon_{m,n} & = 8 \exp\left\{ \frac{\alpha-1}{2\alpha} \left[ |C||E|^m \log(n+1) +
n H_{\widetilde{\alpha}}(A | B^m)_{\Psi}
+ n \log \frac{|R|}{|B_1|} \right] \right\}, \\
\vartheta_{m,n} & = 8 \exp\left\{ \frac{\alpha-1}{2\alpha} \Big[ |S| \log(mn+1) - n H_{\alpha}(A|S^m)_{\Psi}
+ n \log (|R||B_1|) \Big] \right\},
\end{align}
and for appropriately chosen $|R|$ and $|B_1|$ as per \eqref{yae68} and \eqref{yae69},
$\beta_{mn}$ decays exponentially in $mn$. Bob now applies the
POVM given by $\{ V_{x^\prime} \cdot (\Phi^{B_2 R})^{\otimes mn} \}$, $x^\prime \in \cX$, and
\begin{multline}
\Pr\{ x | x \} = \tr \rho_x^{B_2^{mn} R^{mn}} \left[ V_x \cdot (\Phi^{B_2 R})^{\otimes mn} \right]
= F \left[ \rho_x^{B_2^{mn} R^{mn}}, V_x \cdot (\Phi^{B_2 R})^{\otimes mn} \right]^2 \geq 1 -  \beta_{mn},
\end{multline}
where the inequality follows from the Fuchs-van de Graaf inequalities between trace distance
and Fidelity \cite{fuchs-graaf-ineq-1998} and in particular Corollary 9.3.2 in Ref. \cite{wilde-book},
and hence, the error of the protocol is upper bounded by $\beta_{mn}$. Lastly, it is easy to show that
a cptp map followed by a POVM can be implement just by a suitably chosen POVM, and hence,
the decoder of the father protocol and the POVM of the superdense coding protocol can be
implemented by a POVM. The claim now follows readily.
\end{proof}

\section{Quantum state redistribution (QSR)}
\label{srdist}

\begin{definition}
A $(\Psi, \err, n)$ QSR protocol consists of
$n$ copies of a pure state $\ket{\Psi}^{ACBR}$ shared between with Alice ($A$ and $C$),
Bob ($B$), and the reference ($R$) unavailable
to both Alice and Bob, a MES $\Phi^{A_1 B_1}$ shared between Alice ($A_1$) and Bob ($B_1$),
Alice applying
$\cE : A_1 C^n A^n \to C_2 C_3 \widetilde{A}^n$, a quantum communication across a noiseless quantum
channel from Alice to Bob $\cI^{C_3 \to \widetilde{B}}$, and Bob applying
$\cD: B_1 \widetilde{B} B^n \to B_2 \widetilde{B}_3^n B_3^n$ such that for
\beq
\rho^{C_2 B_2 \widetilde{A}^n \widetilde{B}_3^n B_3^n R^n} \equiv
\cD^{B_1 \widetilde{B} B^n \to B_2 \widetilde{B}_3^n B_3^n}
\circ \cI^{C_3 \to \widetilde{B}} \circ \cE^{A_1 C^n A^n \to C_2 C_3 \widetilde{A}^n}
[(\Psi^{ACBR})^{\otimes n} \otimes \Phi^{A_1 B_1}],
\enq
\beq
\left\| \rho^{C_2 B_2 \widetilde{A}^n \widetilde{B}_3^n B_3^n R^n} - \Phi^{C_2 B_2} \otimes
(\Psi^{\widetilde{A} \widetilde{B}_3 B_3 R})^{\otimes n} \right\|_1 \leqslant \err.
\enq
The number $(\log |C_3|)/n$ is called the {\bf quantum communication rate}
and \linebreak $(\log |B_1| - \log |C_2|)/n$ is called the {\bf entanglement cost rate} of the protocol.

A pair of real numbers $(\cR_Q,\cR_E)$ is called an {\bf achievable rate pair} if there
exist, for $n \to \infty$, QSR protocols with quantum communication rate approaching
$\cR_Q$, entanglement cost rate approaching $\cR_E$, and $\err$ approaching $0$.
\end{definition}

The achievable rates are described by the following theorem.

\begin{theorem}[Devetak and Yard, 2008 \cite{devetak-state-redist-2008}]
\label{yatheorem3}
The following rates are achievable for the QSR protocol:
\beq
\cR_Q > \frac{1}{2} I(C:R|B)_\Psi \hspace{0.5in} \text{and} \hspace{0.5in}
\cR_Q + \cR_E > H(C|B)_\Psi.
\enq
\end{theorem}

Our goal in the remainder of this section is to provide the random coding exponents
for the achievability of this protocol.

\begin{theorem}
For any $n \in \mathbb{N}$, there exists a $(\Psi, \err, n)$ QSR protocol for any $\alpha \in (1,2]$,
$\delta_1, \delta_2 > 0$, such that
\begin{align}
\frac{\log |C_3|}{n} & = \frac{1}{2} \left[ H_{\widetilde{\alpha}}(C | B)_{\Psi} - H_\alpha(C|BR)_\Psi \right]
+ (|A| + |B|) |R| \frac{ \log(n+1) }{2n} + \frac{\delta_1 + \delta_2}{2}, \\
\frac{1}{n} \log \frac{|C_3| |B_1|}{|C_2|}
& = H_{\widetilde{\alpha}}(C | B)_{\Psi} + |A| |R| \frac{ \log(n+1) }{n} + \delta_2,
\end{align}
and the error approaches $0$ exponentially in $n$.
\end{theorem}
\begin{proof}
Our line of attack is similar to that in Ref. \cite{qsr-pra-2008}.
Let $W^{C^n \to B_1 C_2 C_3}$, $|B_1| |C_2| |C_3| \leqslant |C|^n$, be a full-rank partial isometry.
Then we can claim using Corollary \ref{corollary1} that for any $\alpha \in (1,2]$, there exists
a Unitary $U^{C^n}$ such that
\begin{multline}
\label{yae48}
\Big\| \tr_{C_2 C_3} \circ \mapone^{C^n \to B_1 C_2 C_3} \Big[ U^{C^n} \cdot
(\Psi^{C B R})^{\otimes n} \Big] - \pi^{B_1} \otimes (\Psi^{BR})^{\otimes n} \Big\|_1 \\
\leqslant 8 \exp\left\{ \frac{\alpha-1}{2\alpha} \left[ |B| |R| \log(n+1) - n H_{\alpha}(C | B R)_{\Psi}
+ \log \frac{|B_1|}{|C_2| |C_3|} \right] \right\} \equiv \varepsilon_n,
\end{multline}
and
\begin{multline}
\label{yae49}
\Big\| \tr_{B_1 C_3} \circ \mapone^{C^n \to B_1 C_2 C_3} \Big[ U^{C^n} \cdot
(\Psi^{\widetilde{A} C R})^{\otimes n} \Big] - \pi^{C_2} \otimes
(\Psi^{\widetilde{A} R})^{\otimes n} \Big\|_1 \\
\leqslant 8 \exp\left\{ \frac{\alpha-1}{2\alpha} \left[ |A| |R| \log(n+1) - n H_{\alpha}(C | A R)_{\Psi}
+ \log \frac{| C_2 |}{|B_1| |C_3|} \right] \right\} \\
= 8 \exp\left\{ \frac{\alpha-1}{2\alpha} \left[ |A| |R| \log(n+1) + n H_{\widetilde{\alpha}}(C | B)_{\Psi}
+ \log \frac{| C_2 |}{|B_1| |C_3|} \right] \right\} \equiv \vartheta_n.
\end{multline}
Using \eqref{yae48}, we claim that there exists an isometry $V_1^{C_2 C_3 \widetilde{A}^n \to A_1 A^n C^n}$
such that
\beq
\Big\| V_1 \cdot \mapone \Big[ U^{C^n} \cdot
(\Psi^{\widetilde{A} C B R})^{\otimes n} \Big] - \Phi^{A_1 B_1} \otimes (\Psi^{ACBR})^{\otimes n}
\Big\|_1 \leqslant \Xi(\varepsilon_n),
\enq
and hence, using the compressive map $\cC_{V_1^\dag}: A_1 A^n C^n \to C_2 C_3 \widetilde{A}^n$, we have
\begin{align}
\label{yae50}
\Big\| \cC_{V_1^\dag} \left[ \Phi^{A_1 B_1} \otimes (\Psi^{ACBR})^{\otimes n} \right] - \mapone
\Big[ U^{C^n} \cdot (\Psi^{\widetilde{A} C B R})^{\otimes n} \Big] \Big\|_1
\leqslant \Xi(\varepsilon_n).
\end{align}
Using \eqref{yae49}, we claim that there exists an isometry
$V_2^{B_1 \widetilde{B} B^n \to B_2 \widetilde{B}_3^n B_3^n}$ such that
\beq
\Big\| V_2 \cdot \cI \circ \mapone \Big[ U^{C^n} \cdot
(\Psi^{\widetilde{A} C B R})^{\otimes n} \Big] - \Phi^{C_2 B_2} \otimes
(\Psi^{\widetilde{A} \widetilde{B}_3 B_3 R})^{\otimes n}
\Big\|_1 \leqslant \Xi(\vartheta_n).
\enq
Using monotonicity and triangle inequality, we now have
\beq
\Big\| V_2 \cdot \cI \circ \cC_{V_1^\dag} \left[ \Phi^{A_1 B_1} \otimes (\Psi^{ACBR})^{\otimes n} \right] -
\Phi^{C_2 B_2} \otimes (\Psi^{\widetilde{A} \widetilde{B}_3 B_3 R})^{\otimes n} \Big\|_1
\leqslant \Xi(\varepsilon_n) + \Xi(\vartheta_n).
\enq
Hence, Alice's operation is $\cC_{V_1^\dag}^{A_1 A^n C^n \to C_2 C_3 \widetilde{A}^n}$
and Bob's operation is $V_2^{B_1 \widetilde{B} B^n \to B_2 \widetilde{B}_3^n B_3^n}$.
%It is now clear that we have an exponential convergence of the Fidelity of the protocol to $1$ if
%for any $\alpha \in (1,2]$ and $\delta_1, \delta_2 > 0$,
%\begin{align}
%\frac{\log|C_2|}{n} - \frac{\log |B_1|}{n} - \frac{\log |C_3|}{n} & = H_\alpha(C|BR)_\Psi -
%|B| |R| \frac{ \log(n+1) }{n} - \delta_1 \\
%\frac{\log|B_1|}{n} - \frac{\log |C_2|}{n} - \frac{\log |C_3|}{n} & = -H_{\widetilde{\alpha}}(C | B)_{\Psi} -
%|A| |R| \frac{ \log(n+1) }{n} - \delta_2.
%\end{align}
%We exhaust the entire achievable rate region as stipulated by Theorem \ref{yatheorem3}.
The claim now follows readily.
\end{proof}

\section{Quantum communication across Broadcast Channels \\ (QCBC)}
\label{qbc}

\begin{definition}
A $(\cN, \err, n)$ QCBC protocol consists of
$n$ copies of four MES $\ket{\Phi}^{S_1 R_1}$, $\ket{\Phi}^{A_1 B_1}$,
$\ket{\Phi}^{S_2 R_2}$, and $\ket{\Phi}^{A_2 B_2}$,
where Alice has $S_1$, $S_2$, $A_1$, $A_2$, Bob $1$ has $B_1$, Bob $2$ has $B_2$,
and the references ($R_1$ and $R_2$) are inaccessible to both Alice and Bob,
Alice applying the encoding map $\cE^{A_1^n S_1^n A_2^n S_2^n \to A^{\prime n}}$,
$n$ uses of a quantum broadcast channel from Alice to Bob
$1$ and $2$, $\cN^{A^\prime \to C_1 C_2}$ (with
Stinespring dilation $V_{\cN}^{A^\prime \to C_1  C_2 E}$), and
local quantum operations by Bobs $\cD_i^{B_i^n C_i^n \to \widetilde{S}_i^n}$, $i=1,2$, such that for
\begin{multline}
\rho^{\widetilde{S}_1^n R_1^n \widetilde{S}_2^n R_2^n} \equiv \left(\cD_1^{B_1^n C_1^n \to \widetilde{S}_1^n}
\circ \cD_2^{B_2^n C_2^n \to \widetilde{S}_2^n} \right)
\circ (\cN^{A^\prime \to C_1 C_2})^{\otimes n} \\
\circ \cE^{A_1^n S_1^n A_2^n S_2^n \to A^{\prime n}}
\left[ (\Phi^{S_1 R_1} \otimes \Phi^{A_1 B_1} \otimes \Phi^{S_2 R_2}
\otimes \Phi^{A_2 B_2})^{\otimes n} \right],
\end{multline}
\beq
\left\| \rho^{\widetilde{S}_1^n R_1^n \widetilde{S}_2^n R_2^n} - (\Phi^{\widetilde{S}_1 R_1} \otimes
\Phi^{\widetilde{S}_2 R_2})^{\otimes n} \right\|_1 \leqslant \err.
\enq
For $i=1,2$, the numbers $\log |R_i|$ are the {\bf quantum communication rates}
and $\log |B_i|$ are the {\bf entanglement consumption rates} of the protocol.

A vector of real numbers $(\cR_{Q,1},\cR_{Q,2},\cR_{E,1},\cR_{E_2})$ is called an
{\bf achievable rate vector} if there
exist, for $n \to \infty$, QCBC protocols with quantum communication rates approaching
$\cR_{Q,i}$, entanglement consumption rates approaching $\cR_{E,i}$, $i=1,2$,
and $\err$ approaching $0$.
\end{definition}

\begin{theorem}[Dupuis, 2009 \cite{dupuis-thesis}]
\label{yatheorem8}
Let $\ket{\Psi}^{G_1 G_2 A^{\prime} D}$ be any pure state with
$\ket{\Psi}^{G_1 G_2 C_1 C_2 E D} = V_\cN^{A^\prime \to C_1 C_2 E} \ket{\Psi}^{G_1 G_2 A^{\prime} D}$. The
following rates are achievable:
\begin{align}
\log |R_1| + \log |B_1| & < H(G_1)_\Psi \\
\log |R_2| + \log |B_2| & < H(G_2)_\Psi \\
\log |R_1| + \log |B_1| + \log |R_2| + \log |B_2| & < H(G_1 G_2)_\Psi \\
\log |R_1| - \log |B_1| & < I(G_1 \rangle C_1)_\Psi \\
\log |R_2| - \log |B_2| & < I(G_2 \rangle C_2)_\Psi.
\end{align}
\end{theorem}

%We provide the following theorem for the achievable rates analogous to that provided by
%Dupuis (see Ref. \cite{dupuis-thesis}) but with error decaying exponentially in $n$.

We follow the line of attack in Ref. \cite{dupuis-thesis} that we need to show the
following theorem, which
would yield Theorem \ref{yatheorem8}. The regularized expressions can be obtained by
additional blocking.

\begin{theorem}
For any $n \in \mathbb{N}$, $\ket{\Psi}^{G_1 G_2 A^{\prime} D}$ and
$\ket{\Psi}^{G_1 G_2 C_1 C_2 E D}$ the states defined in Theorem \ref{yatheorem8},
there exists a $(\cN, \err, n)$ QCBC protocol such that for
any $\alpha \in (1,2]$, $\delta_1, \delta_2, \delta_3, \delta_4 > 0$,
\begin{align}
\log |R_1| + \log |B_1| & = H_\alpha(G_1 | G_2)_\Psi -
\frac{|G_2| \log(n+1)}{n} - \delta_1 \\
\log |R_1| - \log |B_1| & = - H_{\widetilde{\alpha}}(G_1 | C_1)_\Psi -
\frac{|G_2 C_2 E D| \log(n+1)}{n} - \delta_2 \\
\log |R_2| + \log |B_2| & = H_\alpha(G_2)_\Psi - \delta_3 \\
\log |R_2| - \log |B_2| & = - H_{\widetilde{\alpha}}(G_2 | C_2)_\Psi -
\frac{|G_1 C_1 E D| \log(n+1)}{n} - \delta_4
\end{align}
and the error approaches $0$ exponentially in $n$.
\end{theorem}
\begin{proof}
Let $W_i^{G_i^n \to R_i^n B_i^n}$, $|G_i|^n \geqslant |R_i|^n |B_i|^n$, $i=1,2$, be full-rank
partial isometries. Define:
\begin{align}
\varepsilon_{n,1} & \equiv
20 \exp\left\{ \frac{\alpha-1}{2\alpha} \Big[ |G_2| \log(n+1) - n H_{\alpha}(G_1 | G_2)_{\Psi}
+ n \log |R_1| |B_1| \Big] \right\} \\
\varepsilon_{n,2} & \equiv
20 \exp\left\{ \frac{\alpha-1}{2\alpha} \left[ |G_2 C_2 E D| \log(n+1) -
n H_{\alpha}(G_1 | G_2 C_2 E D)_{\Psi} + n \log \frac{|R_1|}{|B_1|} \right] \right\} \\
\varepsilon_{n,3} & \equiv 20 \exp\left\{ \frac{\alpha-1}{2\alpha} \Big[ - n H_{\alpha}(G_2)_{\Psi}
+ n \log (|R_2||B_2|) \Big] \right\} \\
\varepsilon_{n,4} & \equiv 20 \exp\left\{ \frac{\alpha-1}{2\alpha} \left[ |G_1 C_1 E D| \log(n+1) -
n H_{\alpha}(G_2 | G_1 C_1 E D)_{\Psi} + n \log \frac{|R_2|}{|B_2|} \right] \right\} \\
\varepsilon_{n,5} & \equiv 20 \exp\left\{ \frac{\alpha-1}{2\alpha} \left[ - n H_{\alpha}(G_2)_{\Psi} -
n \log |G_2| \right] \right\}.
\end{align}
For $i = 1,2$, let $U_i$ be random Unitaries on $G_i^n$. We have
\begin{multline}
\label{yae34}
\Exp_{U_1, U_2} \left\| \cT_{W_1}^{G_1^n \to R_1^n B_1^n} \circ \cT_{W_2}^{G_2^n \to R_2^n B_2^n}
\left[ (U_1 \otimes U_2) \cdot (\Psi^{G_1 G_2})^{\otimes n} \right]
- (\pi^{R_1 B_1 R_2 B_2})^{\otimes n}
\right\|_1 \\
\leqslant \Exp_{U_1, U_2} \left\| \cT_{W_2}^{G_2^n \to R_2^n B_2^n} \left( U_2 \cdot
\left\{ \cT_{W_1}^{G_1^n \to R_1^n B_1^n}
\left[ U_1 \cdot (\Psi^{G_1 G_2})^{\otimes n} \right] -
(\pi^{R_1 B_1})^{\otimes n} \otimes (\Psi^{G_2})^{\otimes n} \right\} \right) \right\|_1 \\
+ \Exp_{U_2} \left\| (\pi^{R_1 B_1})^{\otimes n} \otimes \cT_{W_2}^{G_2^n \to R_2^n B_2^n} \left[
U_2 \cdot (\Psi^{G_2})^{\otimes n} \right]  - (\pi^{R_1 B_1 R_2 B_2})^{\otimes n} \right\|_1 \\
\leqslant  \Exp_{U_1} \left\| \cT_{W_1}^{G_1^n \to R_1^n B_1^n}
\left[ U_1 \cdot (\Psi^{G_1 G_2})^{\otimes n} \right] -
(\pi^{R_1 B_1})^{\otimes n} \otimes (\Psi^{G_2})^{\otimes n} \right\|_1 \hspace{1.3in} \\
+ \Exp_{U_2} \left\| \cT_{W_2}^{G_2^n \to R_2^n B_2^n} \left[
U_2 \cdot (\Psi^{G_2})^{\otimes n} \right]  - (\pi^{R_2 B_2})^{\otimes n} \right\|_1 \\
\leqslant (\varepsilon_{n,1} + \varepsilon_{n,3})/5,
\end{multline}
where the first inequality follows from the triangle's inequality and
the second inequality follows since $\cT_{W_2}$ is a class-$1$ map and the last inequality from
Theorem \ref{theorem1}.
We also have
\begin{multline}
\label{yae35}
\Exp_{U_1, U_2} \left\|
\cT_{W_2} \left(U_2 \cdot \left\{ \tr_{B_1} \circ
\cT_{W_1} \left[ U_1 \cdot
(\Psi^{G_1 G_2 C_2 E D})^{\otimes n} \right] -
(\pi^{R_1} \otimes \Psi^{G_2 C_2 E D})^{\otimes n} \right\} \right) \right\|_1 \\
\leqslant \Exp_{U_1} \left\| \tr_{B_1} \circ
\cT_{W_1} \left[ U_1 \cdot
(\Psi^{G_1 G_2 C_2 E D})^{\otimes n} \right] -
(\pi^{R_1} \otimes \Psi^{G_2 C_2 E D})^{\otimes n} \right\|_1
\leqslant \varepsilon_{n,2}/5,
\end{multline}
where the first inequality follows since $\cT_{W_2}$ is a class-$1$ map,
\begin{align}
\label{yae37}
\Exp_{U_2} \left\| \tr_{B_2} \circ \cT_{W_2} \left[
U_2 \cdot (\Psi^{G_1 G_2 C_1 E D})^{\otimes n} \right] -
(\pi^{R_2} \otimes \Psi^{G_1 C_1 E D})^{\otimes n} \right\|_1 & \leqslant \varepsilon_{n,4}/5, \\
\Exp_{U_2} \left| \tr \circ \cT_{W_2} \left[
U_2 \cdot (\Psi^{G_2})^{\otimes n} \right] - 1 \right| & \leqslant \varepsilon_{n,5}/5.
\end{align}
We now use the arguments in Corollary \ref{corollary1} to claim that there exist
Unitaries $U_i$ on $G_i^n$, $i=1,2$, such that 
\begin{align}
\left\| \cT_{W_1} \circ \cT_{W_2}
\left[ (U_1 \otimes U_2) \cdot (\Psi^{G_1 G_2})^{\otimes n} \right]
- (\pi^{R_1 B_1 R_2 B_2})^{\otimes n} \right\|_1 & \leqslant \varepsilon_{n,1} + \varepsilon_{n,3} \\
\left\| \cT_{W_2} \left( U_2 \cdot \left\{
\tr_{B_1} \circ \cT_{W_1} \left[
U_1 \cdot (\Psi^{G_1 G_2 C_2 E D})^{\otimes n} \right] -
(\pi^{R_1} \otimes \Psi^{G_2 C_2 E D})^{\otimes n} \right\} \right) \right\|_1
& \leqslant \varepsilon_{n,2} \\
\left\| \tr_{B_2} \circ \cT_{W_2} \left[
U_2 \cdot (\Psi^{G_1 G_2 C_1 E D})^{\otimes n} \right] -
(\pi^{R_2} \otimes \Psi^{G_1 C_1 E D})^{\otimes n} \right\|_1 & \leqslant \varepsilon_{n,4} \\
\left| \tr \circ \cT_{W_2} \left[
U_2 \cdot (\Psi^{G_2})^{\otimes n} \right] - 1 \right| & \leqslant \varepsilon_{n,5}.
\end{align}
It now follows that there exist isometries $V_1^{S_1^n A_1^n S_2^n A_2^n \to A^{\prime n} D^n}$,
$V_2^{B_1^n C_1^n \to
\widetilde{G}_1^n \widetilde{S}_1^n \widetilde{C}_1^n}$, and \linebreak
$V_3^{B_2^n C_2^n \to \widetilde{G}_2^n \widetilde{S}_2^n \widetilde{C}_2^n}$ such that
\begin{multline}
\label{yae38}
\Big\| \cT_{W_1} \circ \cT_{W_2}
\left[ (U_1 \otimes U_2) \cdot (\Psi^{G_1 G_2 A^\prime D})^{\otimes n} \right] - V_1 \cdot
(\Phi^{S_1 R_1} \otimes \Phi^{A_1 B_1} \otimes \Phi^{S_2 R_2}
\otimes \Phi^{A_2 B_2})^{\otimes n} \Big\|_1 \\
\leqslant \Xi(\varepsilon_{n,1} + \varepsilon_{n,3}),
\end{multline}
\begin{multline}
\label{yae39}
\left\| \cT_{W_2} \left( U_2 \cdot \left\{ V_2 \circ \cT_{W_1}  \left[ U_1 \cdot
(\Psi^{G_1 G_2 C_1 C_2 E D})^{\otimes n}
\right] - (\Phi^{R_1 \widetilde{S}_1})^{\otimes n} \otimes
(\Psi^{\widetilde{G}_1 G_2 \widetilde{C}_1 C_2 E D})^{\otimes n} \right\} \right) \right\|_1 \\
\leqslant \Xi(\varepsilon_{n,2}) + \varepsilon_{n,5},
\end{multline}
where we have used the triangle's inequality, and
\begin{multline}
\label{yae40}
\left\| V_3 \cdot \cT_{W_2}^{G_2^n \to R_2^n B_2^n} \left[ U_2 \cdot
(\Psi^{G_1 G_2 C_1 C_2 E D})^{\otimes n}
\right] - (\Phi^{R_2 \widetilde{S}_2})^{\otimes n} \otimes
(\Psi^{G_1 \widetilde{G}_2 C_1 \widetilde{C}_2 E D})^{\otimes n} \right\|_1 \\
\leqslant \Xi(\varepsilon_{n,4}).
\end{multline}
We now have for
\begin{align}
\cE^{A_1 S_1 A_2 S_2 \to A^{\prime n}} & \equiv \tr_{D^n}
\circ V_1^{S_1^n A_1^n S_2^n A_2^n \to A^{\prime n} D^n}, \\
\cD_1^{B_1^n C_1^n \to \widetilde{S}_1^n} & \equiv \tr_{
\widetilde{G}_1^n \widetilde{C}_1^n} \circ V_2^{B_1^n C_1^n \to
\widetilde{G}_1^n \widetilde{S}_1^n \widetilde{C}_1^n}, \\
\cD_2^{B_2^n C_2^n \to \widetilde{S}_2^n} & \equiv \tr_{
\widetilde{G}_2^n \widetilde{C}_2^n} \circ V_3^{B_2^n C_2^n \to
\widetilde{G}_2^n \widetilde{S}_2^n \widetilde{C}_2^n}, \\
\Upsilon^{R_1^n \widetilde{S}_1^n R_2^n \widetilde{S}_2^n} &
\equiv \cD_1 \circ \cD_2 \circ \cT_{W_1} \circ \cT_{W_2}
\left[ (U_1 \otimes U_2) \cdot (\Psi^{G_1 G_2 C_1 C_2})^{\otimes n} \right], \\
\Upsilon^{R_1^n \widetilde{S}_1^n R_2^n \widetilde{S}_2^n}_2 & \equiv
(\Phi^{R_1 \widetilde{S}_1})^{\otimes n} \otimes
\cD_2 \circ \cT_{W_2} \left[ U_2 \cdot (\Psi^{G_2 C_2})^{\otimes n} \right],
\end{align}
\begin{multline}
\label{yae41}
\Big\| \cD_1 \circ \cD_2 \circ (\cN)^{\otimes n} \circ \cE
\left[ (\Phi^{S_1 R_1} \otimes \Phi^{A_1 B_1} \otimes \Phi^{S_2 R_2}
\otimes \Phi^{A_2 B_2})^{\otimes n} \right] -
(\Phi^{R_1 \widetilde{S}_1} \otimes \Phi^{R_2 \widetilde{S}_2})^{\otimes n} \Big\|_1 \\
\leqslant \left\| \Upsilon^{R_1^n \widetilde{S}_1^n R_2^n \widetilde{S}_2^n} -
(\Phi^{R_1 \widetilde{S}_1} \otimes \Phi^{R_2 \widetilde{S}_2})^{\otimes n} \right\|_1
+ \Xi(\varepsilon_{n,1} + \varepsilon_{n,3}) \hspace{2.3in} \\
\leqslant \left\| \Upsilon^{R_1^n \widetilde{S}_1^n R_2^n \widetilde{S}_2^n}
- \Upsilon^{R_1^n \widetilde{S}_1^n R_2^n \widetilde{S}_2^n}_2 \right\|_1
+ \left\| \Upsilon^{R_1^n \widetilde{S}_1^n R_2^n \widetilde{S}_2^n}_2 -
(\Phi^{R_1 \widetilde{S}_1} \otimes \Phi^{R_2 \widetilde{S}_2})^{\otimes n} \right\|_1 +
\Xi(\varepsilon_{n,1} + \varepsilon_{n,3}) \\
\leqslant \Xi(\varepsilon_{n,2}) + \varepsilon_{n,5} +
\Xi(\varepsilon_{n,4}) + \Xi(\varepsilon_{n,1} + \varepsilon_{n,3}),
\end{multline}
where the first inequality follows from \eqref{yae38}, the triangle inequality and the monotonicity,
the second inequality follows from the triangle inequality, the third inequality follows from
\eqref{yae39}, \eqref{yae40}, and monotonicity.
The claim of the Theorem now follows from \eqref{yae41}.
\end{proof}

\noindent {\bf Remark}:
It is clear from the above theorem that any rate in the following rate region is achievable with
error decaying exponentially in $n$ to zero:
\begin{align}
\log |R_1| + \log |B_1| & < H(G_1|G_2)_\Psi \\
\log |R_2| + \log |B_2| & < H(G_2)_\Psi \\
\log |R_1| - \log |B_1| & < I(G_1 \rangle C_1)_\Psi \\
\log |R_2| - \log |B_2| & < I(G_2 \rangle C_2)_\Psi
\end{align}
We now repeat the argument given in Theorem 5.3 in Ref. \cite{dupuis-thesis} that by
switching the roles of Bob 1 and Bob 2 and doing time sharing, we can achieve any point in
the rate region as stipulated in the claim of Theorem \ref{yatheorem8}.

\section{Destroying correlations by adding classical randomness}

\begin{definition}
A $(\rho,\err,n)$ protocol for destroying correlations by adding classical randomness
consists of $n$ copies of a bipartite state $\rho^{AR}$, and
applying $M$ Unitaries $U_i$, $i=1,...,M$, over $A ^n$ such that
\beq
\Big\| \frac{1}{M} \sum_{i=1}^M \left[ U_i \cdot (\rho^{AR})^{\otimes n} \right] -
\sigma^{A^n} \otimes (\rho^R)^{\otimes n} \Big\|_1 \leqslant \err,
\enq
where $\sigma^{A^n} \in \densitymatrix(\cH_{A^n})$ and we make no apriori restrictions on
the choice of $\sigma^{A^n}$.

The number $(\log M)/n$ is called the {\bf rate} of the protocol.
A real numbers $\cR_C$ is called an {\bf achievable rate} if there
exist, for $n \to \infty$, protocols with rate approaching $\cR_C$ and the $\err$ approaching $0$.
\end{definition}

\begin{theorem}[Groisman \emph{et al}, 2005 \cite{corr-pra-2005}]
The smallest achievable rate is $I(A:R)_\rho$.
\end{theorem}

We prove the following theorem.

\begin{theorem}
For any $n \in \mathbb{N}$, there exists a $(\rho,\err,n)$ protocol such that
for any $\delta > 0$, $\alpha \in (1,2]$ and $\ket{\Psi}^{ARE}$ a purification of $\rho^{AR}$,
\beq
\frac{\log M}{n} = H_{\widetilde{\alpha}}(A)_\rho - H_{\alpha}(A | R)_\rho +
(|E|+1)|R| \frac{\log(n+1)}{n} + \delta,
\enq
and the error approaches $0$ exponentially in $n$.
\end{theorem}
\begin{proof}
Consider a partial isometry $W^{A^n \to B}$, $|B| \leqslant |A^n|$.
For $M \leqslant |B|^2$, we can choose
$M$ Unitaries $V_i^B \in \bbU(B)$ such that
$\tr (V_i^B)^\dag V_j^B = |B| \delta_{i,j}$, and
let $\cV_M : B \to B$ be a cptp map given by
\beq
\cV_M(\sigma^B) \equiv \frac{1}{M} \sum_{i=1}^M V_i^B \cdot \sigma^B.
\enq
Then, from Corollary \ref{corollary1}, for any $\alpha \in (1,2]$, there exists a Unitary $U$ such that
\begin{multline}
\label{yae51}
\left\| \tr_B \circ \mapone \left[ U \cdot (\Psi^{ARE})^{\otimes n} \right] - (\Psi^{RE})^{\otimes n} \right\|_1
\\ \leqslant 8 \exp\Big\{ \frac{\alpha-1}{2\alpha} \big[ |R| |E| \log(n+1) +
n H_{\widetilde{\alpha}}(A)_{\rho} - \log|B| \big] \Big\} \equiv \varepsilon_n,
\end{multline}
and
\begin{multline}
\label{yae52}
\left\| \cV_M \circ \mapone \left[ U \cdot (\rho^{AR})^{\otimes n} \right] - \pi^B \otimes (\rho^R)^{\otimes n}
\right\|_1 \\
\leqslant 8 \exp\Big\{ \frac{\alpha-1}{2\alpha} \big[ |R| \log(n+1) -
n H_\alpha(A|R)_{\rho} - \log M + \log|B| \big] \Big\} \equiv \vartheta_n,
\end{multline}
where we have used $\Theta(\cV_M \circ \mapone) \leq \log |B| - \log M$ from Lemma \ref{yal12}.
From \eqref{yae51} and Lemma \ref{yal11}, we claim that there exists a Unitary $U_2$ over $A^n$
such that
\beq
\label{yae53}
\left\| W^\dag \cdot \mapone \left[ U \cdot (\Psi^{ARE})^{\otimes n} \right] -
U_2 \cdot (\Psi^{ARE})^{\otimes n} \right\|_1
\leqslant \Xi(\varepsilon_n).
\enq
Consider now the following Unitaries over $A^n$ constructed from $V_i^B$ as
$V_i^{A^n} = W^\dag \cdot V_i^B + (\eye^A - W^\dag W)$. Note that $V_i^{A^n} W^\dag = W^\dag V_i^B$.
We now claim that $V_i^{A^n} U_2$ are the $M$ Unitaries we need. We have
\begin{align}
\Bigg\| \frac{1}{M} \sum_{i=1}^M (V_i^{A^n} & U_2) \cdot (\rho^{AR})^{\otimes n}
- (W^\dag \cdot \pi^B) \otimes (\rho^R)^{\otimes n} \Bigg\|_1 \nonumber \\
& \leqslant \left\| \frac{1}{M} \sum_{i=1}^M (V_i^{A^n} U_2) \cdot (\rho^{AR})^{\otimes n}
- \frac{1}{M} \sum_{i=1}^M (V_i^{A^n} W^\dag) \cdot \mapone \left[ U \cdot (\rho^{AR})^{\otimes n} \right] \right\|_1
+ \nonumber \\
& \hspace{0.43in} \left\| \frac{1}{M} \sum_{i=1}^M (V_i^{A^n} W^\dag) \cdot \mapone
\left[ U \cdot (\rho^{AR})^{\otimes n} \right]
- (W^\dag \cdot \pi^B) \otimes (\rho^R)^{\otimes n} \right\|_1 \\
& \leqslant \frac{1}{M} \sum_{i=1}^M \Big\|  (V_i^{A^n} U_2) \cdot (\rho^{AR})^{\otimes n}
- (V_i^{A^n} W^\dag) \cdot \mapone \left[ U \cdot (\rho^{AR})^{\otimes n} \right] \Big\|_1 + \nonumber \\
& \hspace{0.43in}
\left\| \frac{1}{M} \sum_{i=1}^M (W^\dag V_i^B) \cdot \mapone \left[ U \cdot (\rho^{AR})^{\otimes n} \right]
- (W^\dag \cdot \pi^B) \otimes (\rho^R)^{\otimes n} \right\|_1 \\
& \leqslant \frac{1}{M} \sum_{i=1}^M \Big\|  U_2 \cdot (\rho^{AR})^{\otimes n}
- W^\dag \cdot \mapone \left[ U \cdot (\rho^{AR})^{\otimes n} \right] \Big\|_1
+ \nonumber \\
& \hspace{0.43in}
\left\| \frac{1}{M} \sum_{i=1}^M V_i^B \cdot \mapone \left[ U \cdot (\rho^{AR})^{\otimes n} \right]
- \pi^B \otimes (\rho^R)^{\otimes n} \right\|_1 \\
& \leqslant \Xi(\varepsilon_n) + \vartheta_n, 
\end{align}
where the first inequality follows from the triangle inequality, in the second inequality, the
first term follows from the convexity of the trace norm and the second term follows by invoking
$V_i^{A^n} W^\dag = W^\dag V_i^B$, in the third inequality, the first term follows by invoking the
Unitary invariance of the trace norm and the second term from monotonicity,
in the fourth inequality, the first term is upper bounded
using \eqref{yae53} and the second term is upper bounded using \eqref{yae52}. The claim now follows readily.
\end{proof}

\section{Conclusions}

In conclusion, we have provided a new version of the decoupling theorem that
gives an exponential bound on the average decoupling error with a \renyi $\alpha$-conditional
entropy in the exponent for a restricted class of completely positive maps
for any $\alpha \in (1,2]$ as opposed to only $\alpha = 2$ in Ref. \cite{dupuis-thesis}.
This key step allows us to make a connection with the random coding exponents, which we provide for
several important protocols including those at the top of the family tree of protocols. The importance of random
coding exponents for the achievability of information-processing tasks has been well known
since the seminal work by Gallager \cite{gallager-expo-1965}. Such an analysis, with
very few exceptions thus far, has been missing and we now fill
that void with this paper. The version of the decoupling theorem and other ideas developed in this paper may well find wider
applications with or without further extensions.

\appendix

\section{Computation of $\Theta$ for some cases}

\begin{lemma}
For a full-rank partial isometry $W^{A \to A_1 A_2}$, $|A_1| |A_2| \leqslant |A|$,
\label{yal5}
\begin{align}
\Theta(\tr_{A_2} \circ \mapone^{A \to A_1 A_2}) & \leqslant \log \frac{|A_1|}{|A_2|} \\
\Theta(\tr_{A_2} \circ \cC_W^{A \to A_1 A_2}) & \leqslant \log \frac{|A_1|}{|A_2|}.
\end{align}
\end{lemma}
\begin{proof}
Since we have the freedom in choosing the local orthonormal bases in describing the
MES, hence, let them be such that $W \ket{i}^A = \ket{i}^{A_1 A_2}$ for $i \leqslant |A_1| |A_2|$,
and $W \ket{i}^A = 0$ for $i > |A_1| |A_2|$, where $\{ \ket{i}^A \}$ and $\{ \ket{i}^{A_1 A_2} \}$
are orthonormal states in their respective systems.
It now follows that
\begin{align}
\mapone(\ket{i} \bra{j}^A) & = \frac{|A|}{|A_1| |A_2|} \ket{i} \bra{j}^{A_1 A_2} \ind_{\{i,j \leqslant |A_1||A_2|\}} \\
\cC_W(\ket{i} \bra{j}^A) & = \ket{i} \bra{j}^{A_1 A_2} \ind_{\{i,j \leqslant |A_1||A_2|\}} +
\delta_{i,j} \pi^{A_1 A_2} \ind_{\{i,j > |A_1||A_2|\}}.
\end{align}
We now have
\begin{align}
\exp\{ \Theta(\tr_{A_2} \circ \cC_W) \} & \leqslant \frac{|A_1|}{|A|^2}
\sum_{i,j} \, \tr  \left[ \tr_{A_2} \circ \cC_W(\ket{i} \bra{j}^A) \right]
\left[ \tr_{A_2} \circ \cC_W(\ket{j} \bra{i}^A) \right] \\
& = \frac{|A_1|}{|A|^2} \, \tr \Bigg[ \sum_{i,j \leqslant |A_1| |A_2|} \tr_{A_2} ( \ket{i} \bra{j}^{A_1 A_2} )
\tr_{A_2} ( \ket{j} \bra{i}^{A_1 A_2} ) + \nonumber \\
& \hspace{1.0in} 
\sum_{i,j > |A_1| |A_2|} \delta_{i,j} \left( \tr_{A_2} \pi^{A_1 A_2} \right)^2 \Bigg] \\
& = \frac{|A_1|}{|A|^2} \left( |A_1|^2 |A_2| + \frac{|A| - |A_1| |A_2|}{|A_1|} \right)
\leqslant \frac{|A_1|}{|A_2|},
\end{align}
where the first inequality follows using \eqref{yae8}. Following the above, we arrive
at
\beq
\exp\{ \Theta(\tr_{A_2} \circ \mapone) \} \leqslant \frac{|A_1|}{|A|^2} \left[ \left( \frac{|A|}{|A_1| |A_2|} \right)^2
|A_1|^2 |A_2| \right] = \frac{|A_1|}{|A_2|}.
\enq
QED.
\end{proof}

\begin{lemma}
\label{yal8}
Let $\{M_i \in \mathrm{L}(BC, D),
i=1,...,J\}$, $J = \lceil \frac{BC}{D} \rceil$, be a complete set of measurement operators
($\sum_i M_i^\dagger M_i = \eye^{BC}$). Let $\zeta \equiv \frac{|B||C|}{|D|}$ and let
the first $\vartheta \equiv \lfloor \frac{BC}{D} \rfloor$ $M_i$'s be rank-$|D|$ partial isometries.
Define for any orthonormal basis $\{\ket{i}^{X}\}$, $i=1,...,J$,
\beq
\cE^{BC \to X D}(\sigma^{BC}) = \sum_{i=1}^J \ketbra{i}^{X} \otimes (M_i \cdot \sigma^{BC})
\enq
and let $W^{A \to B}$, $|B| \leqslant |A|$, be a full-rank partial isometry.
Then
\beq
\Theta(\cE \circ \mapone) \leqslant \log |D|.
\enq
\end{lemma}
\begin{proof}
Let $W^{A \to B} = \sum_{i=1}^{|B|} \ket{i}^B \bra{i}^A$.
Once again, we exploit the freedom in choosing the local bases in defining MES and have
\beq
\ket{\Phi}^{AA^\prime C C^\prime} = \frac{1}{\sqrt{|A||C|}} \sum_{i_1,i_2} \ket{i_1}^A \ket{i_1}^{A^\prime}
\ket{i_2}^C \ket{i_2}^{C^\prime}.
\enq
Hence,
\beq
X_{i_1,j_1} \equiv \mapone (\ket{i_1} \bra{j_1}^A) = \frac{|A|}{|B|} \ket{i_1} \bra{j_1}^B \ind_{\{i_1,j_1 \leqslant |B|\}}.
\enq
We now have for $\theta^{XD} = \sum_x p_x \ketbra{x}^X \otimes \pi^D$, $\{p_x\}$ a probability
vector (whose choice is specified below),
\begin{align}
\exp & \{\Theta(\cE \circ \mapone)\} \nonumber \\
& \leqslant \frac{1}{|A|^2 |C|^2}
\sum_{i_1,j_1,i_2,j_2} \tr \left[ \cE ( X_{i_1,j_1} \otimes \ket{i_2} \bra{j_2}^C )
\cE ( X_{j_1,i_1} \otimes \ket{j_2} \bra{i_2}^C ) \right] (\theta^{XD})^{-1} \\
& = \frac{1}{|B|^2 |C|^2}
\sum_{i_1,j_1,i_2,j_2,x} \tr \Big[ \ketbra{x}^X \otimes
M_x ( \ket{i_1} \bra{j_1}^B \otimes \ket{i_2} \bra{j_2}^C ) M_x^\dag M_x
\nonumber \\
& \hspace{2in}  ( \ket{j_1} \bra{i_1}^B \otimes \ket{j_2} \bra{i_2}^C ) M_x^\dag \Big] (\theta^{XD})^{-1} \\
& = \frac{1}{|B|^2 |C|^2}
\sum_{i_1, i_2,x} (\tr M_x M_x^\dag) \tr \Big[ \ketbra{x}^X \otimes
M_x ( \ket{i_1} \bra{i_1}^B \otimes \ket{i_2} \bra{i_2}^C ) M_x^\dag
\Big] (\theta^{XD})^{-1} \\
& = \frac{1}{|B|^2 |C|^2}
\sum_x (\tr M_x M_x^\dag)^2 \frac{|D|}{p_x}.
\end{align}
Let $p = 1- \vartheta/\zeta$,
$p_x = (1-p)/\vartheta$ for $x = 1,...,\vartheta$,
and if $|D|$ doesn't divide $|B||C|$, then there is an additional entry
$p_x = p$ if $x = \vartheta + 1$. Continuing from above, we now have
\begin{align}
\exp \{\Theta(\cE \circ \mapone)\}
& \leqslant \frac{1}{|B|^2 |C|^2} \left[ \vartheta  |D|^2 \frac{|D|}{\frac{1-p}{\vartheta}} +
(|B| |C| - |D| \vartheta)^2 \frac{|D|}{p} \right] \\
& = |D| \left[ \frac{\vartheta^2}{\zeta^2 (1-p) }
+ \left( 1 - \frac{\vartheta}{\zeta} \right)^2 \frac{1}{p} \right] \\
& = |D| \left[ \frac{\vartheta}{\zeta} + 1 - \frac{ \vartheta }{\zeta} \right] = |D|.
\end{align}
QED.
\end{proof}

\begin{lemma}
\label{yal12}
For $M \in \mathbb{N}$, $M \leqslant |B|^2$,
$M$ Unitaries $V_i^B \in \bbU(B)$ such that
$\tr (V_i^B)^\dag V_j^B = |B| \delta_{i,j}$,
let $\cV_M : B \to B$ be a cptp map given by
\beq
\cV_M(\sigma^B) \equiv \frac{1}{M} \sum_{i=1}^M V_i^B \cdot \sigma^B.
\enq
Then
\beq
\Theta(\cV_M \circ \mapone) \leq \log |B| - \log M.
\enq
\end{lemma}
\begin{proof}
Let $W \ket{i}^A = \ket{i}^B$ for $i \leqslant |B|$,
and $W \ket{i}^A = 0$ for $i > |B|$, where $\{ \ket{i}^A \}$ and $\{ \ket{i}^B \}$
are orthonormal states in their respective systems.
Using $\mapone(\ket{i} \bra{j}^A) = \frac{|A|}{|B|} \ket{i} \bra{j}^{B} \ind_{\{i,j \leqslant |B|\}}$, we have
\begin{align}
\exp\{ \Theta(\cV_M \circ \mapone) \} & \leqslant \frac{|B|}{|A|^2}
\sum_{i,j} \, \tr  \left[ \cV \circ \mapone(\ket{i} \bra{j}^A) \right]
\left[ \cV \circ \mapone(\ket{j} \bra{i}^A) \right] \\
& = \frac{1}{|B|} \, \tr \left[ \sum_{i,j \leqslant |B|} \cV ( \ket{i} \bra{j}^{B} ) \,
\cV ( \ket{j} \bra{i}^{B} ) \right] \\
& = \frac{1}{|B| M^2} \, \tr \left[ \sum_{i,j \leqslant |B|} \sum_{k,l=1}^M V_k \ket{i} \bra{j}^{B} V_k^\dag
V_l \ket{j} \bra{i}^{B} V_l^\dag \right] \\
& = \frac{1}{|B| M^2} \, \sum_{k,l=1}^M \left| \tr V_l^\dag V_k \right|^2
= \frac{1}{|B| M^2} \, \sum_{k,l=1}^M |B|^2 \delta_{k,l} = \frac{|B|}{M},
\end{align}
where the first inequality follows using \eqref{yae8} and the fourth equality follows
since $\tr V_l^\dag V_k = |B| \delta_{k,l}$. QED.
\end{proof}

\section{Lemmata}

\begin{lemma}
\label{yal1}
Let $\cT$ be a completely positive map. Then for any inputs $\sigma$, $\theta$
(not necessarily Hermitian), there exists a contraction $K$ such that
\beq
\cT(\sigma \theta^\dagger) \cT(\theta \sigma^\dagger) =
\sqrt{\cT(\sigma \sigma^\dagger) } K \cT(\theta \theta^\dagger) K^\dagger
\sqrt{\cT(\sigma \sigma^\dagger) }.
\enq
In particular, if $\theta = \eye$ and $\cT(\eye)$ is a scaled identity,
i.e., commutes with all matrices, then
\beq
\cT(\sigma) \cT(\sigma^\dagger) \leqslant \cT(\sigma \sigma^\dagger) \cT(\eye).
\enq
An example of such a $\cT$ is the partial trace.
\end{lemma}
\begin{proof}
Since $\cT$ is completely positive, it is also $2$-positive. Hence, if $\cI_2$ is the
identity super-operator for $2 \times 2$ matrices, then for orthonormal $\ket{0}, \ket{1}$,
we have
\begin{align}
0 & \leqslant (\cI_2 \otimes \cT) \left[ (\ket{0} \otimes \theta + \ket{1} \otimes \sigma)
(\ket{0} \otimes \theta + \ket{1} \otimes \sigma)^\dagger \right] \\
& = \ket{0} \bra{0} \otimes \cT(\theta \theta^\dagger) + \ket{1} \bra{0} \otimes
\cT(\sigma \theta^\dagger) + \ket{0} \bra{1} \otimes \cT(\theta \sigma^\dagger) + \ket{1} \bra{1} \otimes \cT(\sigma \sigma^\dagger).
\end{align}
We now invoke Theorem IX.5.9 in Ref. \cite{bhatia-1997} to claim that there exists a
contraction $K$ such that
\beq
\cT(\sigma \theta^\dagger) = \sqrt{\cT(\sigma \sigma^\dagger)} K \sqrt{\cT(\theta \theta^\dagger)}.
\enq
The claim and the particular case now follow easily.
\end{proof}

\begin{lemma}
\label{yal6}
Let $\cT^{A \to E}$ be any completely positive map such that $\tr \, \cT(\eye^A) = |A|$. Then
$\cT^{A \to E}$ is a class-$1$ map. For any cptp map $\cE^{E \to C}$,
$\cE^{E \to C} \circ \cT^{A \to E}$ is also a class-$1$ map.
\end{lemma}
\begin{proof}
Let the Kraus operators of $\cT$ be given by $\{ E_i \}$. We have for a random Unitary $U$ over $A$
and any $\sigma \in \LL(\cH_A)$,
\begin{align}
\Exp_U \| \cT(U \cdot \sigma) \|_1 & = \frac{1}{|A|^2} \sum_j \| \cT(U_j \cdot \sigma) \|_1 \\
& = \Big\| \frac{1}{|A|^2} \sum_{i,j} (\ket{j}^B \otimes E_i U_j) \cdot \sigma \Big\|_1 \\
& = \| \cF(\sigma) \|_1 \\
& \leqslant \| \sigma \|_1,
\end{align}
where in the second equality, $\{ \ket{j}^B \}$ is an orthonormal basis in $B$,
$\cF^{A \to BE}$ is a cptp map with Kraus operators $\{ \frac{1}{|A|} (\ket{j}^B \otimes E_i U_j) \}$,
and the last inequality is well known. The second statement of the claim follows simply
by noting that $\tr \, \cE \circ \cT(\eye^A) = \tr \, \cT(\eye^A) = |A|$. QED.
\end{proof}

\begin{lemma}
\label{yal2}
For any matrices $\sigma^{AR}$, $X^A$, $W^R$ (not necessarily Hermitian) and
for $U$ acting on $A$, we have
\begin{multline}
\Exp_U \Big\{ U \sigma^{AR} U^\dagger (X^{A} \otimes W^R) U (\sigma^{AR})^\dagger
U^\dagger \Big\} \\
= \frac{X^A \otimes \left( |A| \Lambda^R - \Upsilon^R \right) + (\tr X^A) \eye^A \otimes
\left( |A| \Upsilon^R - \Lambda^R \right) }{|A|(|A|^2-1)},
\end{multline}
where
$\Lambda^R \equiv \sigma^R W^R (\sigma^R)^\dagger$ and
$\Upsilon^R \equiv \tr_A \left[ \sigma^{AR} (\eye^A \otimes W^R) (\sigma^{AR})^\dagger \right]$.
\end{lemma}
\begin{proof}
Consider first vectors $\{ \ket{\varphi_i} \}$, $i \in 1,...,6$, in $\cH_A$ and we have
\begin{align}
\Exp_U \Big\{ U & \ket{\varphi_1}\bra{\varphi_2} U^\dagger \ket{\varphi_3}\bra{\varphi_4} U \ket{\varphi_5}
\bra{\varphi_6} U^\dagger \Big\} \\
& = (\eye \otimes \bra{\varphi_4}) \Exp_U \left\{ (U \otimes U) ( \ket{\varphi_1} \ket{\varphi_5} )
( \bra{\varphi_6} \bra{\varphi_2} ) ( U^\dagger \otimes U^\dagger ) \right\}
( \eye \otimes \ket{\varphi_3} ) \\
& = (\eye \otimes \bra{\varphi_4}) \left( \frac{q_1 |A| - q_2}{|A|(|A|^2-1)}
\eye^{A A^\prime} + \frac{q_2 |A| - q_1}{|A| (|A|^2 - 1)} F^{A A^\prime} \right)
( \eye \otimes \ket{\varphi_3} ) \\
& = \frac{q_1 |A| - q_2}{|A|(|A|^2-1)} \bracket{\varphi_4}{\varphi_3} \eye^A +
\frac{q_2 |A| - q_1}{|A| (|A|^2 - 1)} \ket{\varphi_3} \bra{\varphi_4},
\end{align}
where the integral in the second equality is well known
(see Lemma 3.4 in Ref. \cite{dupuis-thesis}),
$q_1 = \bracket{\varphi_6}{\varphi_1} \bracket{\varphi_2}{\varphi_5}$,
$q_2 = \bracket{\varphi_2}{\varphi_1} \bracket{\varphi_6}{\varphi_5}$, and
$F^{AA^\prime}$ is the swap operator.
We have by singular value decomposition:
\begin{align}
X^A = \sum_i \eta_i \ket{y_i} \bra{z_i}^A.
\end{align}
We also have by the singular value and Schmidt decompositions:
\begin{align}
\sigma^{AR} = \sum_{i,j,k} \sqrt{\beta_i^2 \lambda_{i,j} \mu_{i,k}} \ket{v_{ij}} \bra{w_{ik}}^A
\otimes \ket{v_{ij}} \bra{w_{ik}}^R.
\end{align}
Let $i_1^2 = (i_1,i_2)$, $i_1^3 = (i_1,...,i_3)$, $j_1^2 = (j_1,j_2)$ and $k_1^2 = (k_1,k_2)$.
We now have
\begin{align}
\Exp_U & \left\{ U \sigma^{AR} U^\dagger (X^A \otimes W^R) U (\sigma^{AR})^\dagger
U^\dagger \right\} \nonumber \\
& = \sum_{i_1^3,j_1^2,k_1^2} f_1(i_1^3,j_1^2,k_1^2)
\Exp_U \Big\{ U (\ket{v_{i_1 j_1}} \bra{w_{i_1 k_1}}^A
\otimes \ket{v_{i_1 j_1}} \bra{w_{i_1 k_1}}^R) U^\dagger
(\ket{y_{i_3}} \bra{z_{i_3}}^A \otimes W^R) \nonumber \\
& \hspace{2.5in} U (\ket{w_{i_2 j_2}} \bra{v_{i_2 k_2}}^A \otimes
\ket{w_{i_2 j_2}} \bra{v_{i_2 k_2}}^R) U^\dagger \Big\} \\
& = \sum_{i_1^3,j_1^2,k_1^2} f_1(i_1^3,j_1^2,k_1^2) \brackett{w_{i_1 k_1}}{W}{w_{i_2 j_2}}^R
\times \nonumber \\
& \hspace{0.5in} \Exp_U \left\{ U \ket{v_{i_1 j_1}}^A \bra{w_{i_1 k_1}} U^\dagger
\ket{y_{i_3}}^A \bra{z_{i_3}} U \ket{w_{i_2 j_2}}^A \bra{v_{i_2 k_2}}^A U^\dagger \right\}
\otimes \ket{v_{i_1 j_1}} \bra{v_{i_2 k_2}}^R \\
& = \sum_{i_1^3,j_1^2,k_1^2} f_1(i_1^3,j_1^2,k_1^2) \brackett{w_{i_1 k_1}}{W}{w_{i_2 j_2}}^R 
\Big[ \frac{q_1(i_1^2,j_1^2,k_1^2) |A| - q_2(i_1^2,j_1^2,k_1^2)}{|A|(|A|^2-1)}
\bracket{z_{i_3}}{y_{i_3}}^A \eye^A + \nonumber \\
& \hspace{1in}
\frac{q_2(i_1^2,j_1^2,k_1^2) |A| - q_1(i_1^2,j_1^2,k_1^2)}{|A|(|A|^2-1)}
\ket{y_{i_3}} \bra{z_{i_3}}^A \Big]
\otimes \ket{v_{i_1 j_1}} \bra{v_{i_2 k_2}}^R \\
\label{yae1}
& = \frac{X^A \otimes \left( |A| \Lambda^R - \Upsilon^R \right)
+ (\tr X^A) \eye^A \otimes
\left( |A| \Upsilon^R - \Lambda^R \right) }{|A|(|A|^2-1)},
\end{align}
where in the first equality
\beq
f_1(i_1^3,j_1^2,k_1^2) = \sqrt{\beta_{i_1}^2 \lambda_{i_1,j_1} \mu_{i_1,k_1} \eta_{i_3}^2
\beta_{i_2}^2 \lambda_{i_2,j_2} \mu_{i_2,k_2}},
\enq
in the third equality,
\begin{align}
q_1(i_1^2,j_1^2,k_1^2) & = \bracket{w_{i_1 k_1}}{w_{i_2 j_2}}^A
\bracket{v_{i_2 k_2}}{v_{i_1 j_1}}^A \\
q_2(i_1^2,j_1^2,k_1^2) & = \bracket{w_{i_1 k_1}}{v_{i_1 j_1}}^A
\bracket{v_{i_2 k_2}}{w_{i_2 j_2}}^A,
\end{align}
and the fourth equality follows after simplifications. QED.
\end{proof}

\begin{lemma}
\label{yal9}
Let $\cT^{A \to E}$ be a completely positive map with the
Choi-Jamio{\l}kowski representation $\omega^{E A^\prime}_\cT$.
Then for a random Unitary $U$ acting on $A$,
any matrix $\sigma^{AR}$, we have
\begin{multline}
\Exp_U \left\{ \left[ \cT \left( U \cdot \sigma^{AR} \right) - \omega^E_\cT \otimes \sigma^R \right]
\left[ \cT \left( U \cdot \sigma^{AR} \right) - \omega^E_\cT \otimes \sigma^R \right]^\dagger
\right\}
= \frac{\qmap_{A^\prime}(\omega^{E A^\prime}_\cT) \otimes \qmap_A(\sigma^{AR})}{|A|^2-1} \\
\leqslant \frac{|A|^2}{|A|^2-1} \tr_{A^\prime} \left( \omega^{E A^\prime}_\cT \right)^2 \otimes
\tr_A \left[ \sigma^{AR} (\sigma^{AR})^\dagger \right].
\end{multline}
\end{lemma}
\begin{proof}
Let $\cT$ be described by the Kraus operators $\{ T_k \}$. We now have
\begin{align}
\Exp_U & \, \left\{ \left[ \cT \left( U \cdot \sigma^{AR} \right] - \omega^E_\cT \otimes \sigma^R \right)
\left[ \cT \left( U \cdot \sigma^{AR} \right) - \omega^E_\cT \otimes \sigma^R \right]^\dagger
\right\} \nonumber \\
& = \sum_{k,l} T_k \Exp_U \left\{ U \sigma^{AR} U^\dagger T_k^\dagger T_l U (\sigma^{AR})^\dagger
U^\dagger \right\} T_l^ \dagger - \left( \omega^E_\cT \right)^2 \otimes \sigma^R (\sigma^R)^\dagger \\
& = \sum_{k,l} T_k
\Bigg\{ T_k^\dagger T_l \otimes \frac{|A| \sigma^R (\sigma^R)^\dagger -
\tr_A \left[ \sigma^{AR} (\sigma^{AR})^\dagger \right]}{|A|(|A|^2-1)} + \nonumber \\
& \hspace{0.5in} (\tr T_k^\dagger T_l) \eye^A \otimes \frac{
|A| \tr_A \left[ \sigma^{AR} (\sigma^{AR})^\dagger \right]
- \sigma^R (\sigma^R)^\dagger}{|A|(|A|^2-1)} \Bigg\} T_l^ \dagger -
\left( \omega^E_\cT \right)^2 \otimes \sigma^R (\sigma^R)^\dagger \\
& = |A|^2 \left( \omega^E_\cT \right)^2 \otimes \frac{|A| \sigma^R (\sigma^R)^\dagger -
\tr_A \left[ \sigma^{AR} (\sigma^{AR})^\dagger \right]}{|A|(|A|^2-1)} + \nonumber \\
& \hspace{0.4in} |A|^2 \, \tr_{A^\prime} \left( \omega^{E A^\prime}_\cT \right)^2 \otimes
\frac{ |A| \tr_A \left[ \sigma^{AR} (\sigma^{AR})^\dagger \right] - \sigma^R (\sigma^R)^\dagger}{|A|(|A|^2-1)}
- \left( \omega^E_\cT \right)^2 \otimes \sigma^R (\sigma^R)^\dagger \\
& = \left( \omega^E_\cT \right)^2 \otimes \frac{|A|^2 \sigma^R (\sigma^R)^\dagger -
|A| \tr_A \left[ \sigma^{AR} (\sigma^{AR})^\dagger \right]}{|A|^2-1} + \nonumber \\
& \hspace{0.4in} |A| \, \tr_{A^\prime} \left( \omega^{E A^\prime}_\cT \right)^2 \otimes
\frac{ |A| \tr_A \left[ \sigma^{AR} (\sigma^{AR})^\dagger \right] - \sigma^R (\sigma^R)^\dagger}{|A|^2-1}
- \left( \omega^E_\cT \right)^2 \otimes \sigma^R (\sigma^R)^\dagger \\
& = \frac{\qmap_{A^\prime}(\omega^{E A^\prime}_\cT) \otimes \qmap_A(\sigma^{AR})}{|A|^2-1} \\
& \leqslant \frac{|A|^2}{|A|^2-1} \, \tr_{A^\prime} \left( \omega^{E A^\prime}_\cT \right)^2 \otimes
\tr_A \left[ \sigma^{AR} (\sigma^{AR})^\dagger \right],
\end{align}
where in the second equality, we have used Lemma \ref{yal2}, and the inequality
follows by noting from Lemma \ref{yal1} that $|A| \tr_A \left[ \sigma^{AR} (\sigma^{AR})^\dagger \right]
- \sigma^R (\sigma^R)^\dagger$ $\in \mathrm{Pos}(\cH_R)$. QED.
\end{proof}

\begin{lemma}[Exercise 9.9 in Ref. \cite{hayashi}]
\label{yal3}
Let $\rho \in \densitymatrix(\cH_A)$, $\sigma \in \mathrm{Pos}(\cH_A)$, and
$\Pi = \{ \cM_{\sigma}(\rho) \geqslant \zeta \sigma \}$. Then for any $\alpha \in (1,2]$,
we have
\beq
\| \Pi \rho \|_1 \leqslant \zeta^{\frac{1-\alpha}{2}} \sqrt{ Q_{\alpha}(\rho \| \sigma) }
= \zeta^{\frac{1-\alpha}{2}} \, \exp \left\{ \frac{\alpha - 1}{2} D_\alpha(\rho \| \sigma) \right\}.
\enq
\end{lemma}

\begin{lemma}[Hayashi \cite{hayashi}]
\label{yal4}
Let $\rho \in \densitymatrix(\cH_A)$, $\sigma \in \mathrm{Pos}(\cH_A)$,
$\Pi = \{ \cM_{\sigma}(\rho) \geqslant \zeta \sigma \}$
and $\hat{\Pi} = \eye - \Pi$. Then
\beq
\tr \sigma^{-1} \hat{\Pi} \rho^2 \hat{\Pi} \leqslant \nu_\sigma \zeta.
\enq
\end{lemma}
The proof of this lemma is contained in Lemma 9.2 in Ref. \cite{hayashi}.

%\begin{lemma}[Hayden \emph{et al} \cite{hayden-et-al-decoupling-2008}]
%\label{yal13}
%Let $\rho, \sigma \in \densitymatrix(\cH_A)$ and any $\theta \in \bbR$,
%\beq
%\left\| \rho - \theta \sigma \right\|_1 \geqslant \frac{1}{2} \left\| \rho - \sigma \right\|_1.
%\enq
%\end{lemma}

\begin{lemma}
\label{yal14}
Let $\sigma, \rho$ $\in \Pos(\cH_A)$.
Then
\beq
\tr \rho + \tr \sigma - 2 F(\rho,\sigma) \leqslant
\left\| \rho - \sigma \right\|_1 \leqslant \sqrt{ ( \tr \rho + \tr \sigma)^2 - 4 F(\rho,\sigma)^2 }.
\enq
\end{lemma}
\begin{proof}
The proof is essentially along the lines of the Fuchs-van de Graaf inequalities
\cite{fuchs-graaf-ineq-1998}. We know that
\beq
F(\rho,\sigma) = \min_{ \text{POVM} \{ \Lambda_m \} } \sum_m \sqrt{p_m q_m},
\enq
where $p_m \equiv \tr \Lambda_m \rho$ and $q_m \equiv \tr \Lambda_m \sigma$. Note that
$\sum_m p_m = \tr \rho$ and $\sum_m q_m = \tr \sigma$. Let $\{ \Lambda_m\}$
be the minimizing POVM in the above equation. We now have
\begin{multline}
\left\| \rho - \sigma \right\|_1 \geqslant \left\| \sum_m \ketbra{m}^X \otimes \sqrt{\Lambda_m}
\rho \sqrt{\Lambda_m} - \sum_m \ketbra{m}^X \otimes \sqrt{\Lambda_m}
\sigma \sqrt{\Lambda_m} \right\|_1 \\
\geqslant \left\| \sum_m p_m \ketbra{m}^X -
\sum_m q_m \ketbra{m}^X \right\|_1 = \sum_m |p_m - q_m| \\
= \sum_m | \sqrt{p_m} - \sqrt{q_m}| |\sqrt{p_m} + \sqrt{q_m}|
\geqslant \sum_m (\sqrt{p_m} - \sqrt{q_m})^2 \\
= \tr \rho + \tr \sigma - 2 F(\rho,\sigma),
\end{multline}
where the first inequality follows from the monotonicity under the application of a cptp map
with Kraus operators $\{ \ket{m}^X \otimes \sqrt{\Lambda_m} \}$, where $\{ \ket{m}^X \}$ is
an orthonormal basis, and the second inequality follows again from monotonicity under
partial trace.

To prove the other inequality, let $\ket{u_\rho}$ and $\ket{v_\sigma}$ be purifications
of $\rho$ and $\sigma$ respectively such that $F(\rho, \sigma) = \bracket{u_\rho}{v_\sigma}$.
We now have
\beq
\left\| \rho - \sigma \right\|_1 \leqslant \left\| u_\rho - v_\sigma \right\|_1
= \sqrt{ ( \tr \rho + \tr \sigma)^2 - 4 F(\rho,\sigma)^2 }.
\enq
QED.
\end{proof}

\begin{lemma}
\label{yal11}
Let $\Psi^A \in \densitymatrix(\cH_A)$, $\xi^A$ $\in \Pos(\cH_A)$ such that
$\| \xi^A - \Psi^A \|_1 \leqslant \varepsilon$.
Let $\xi^{AB}$, $\Psi^{AC}$, $|B| \leqslant |C|$, be purifications of $\xi^A$ and $\Psi^A$ respectively.
Then there exists a partial isometry $V^{B \to C}$ such that
\beq
\left\| V^{B \to C} \cdot \xi^{AB} - \Psi^{AC} \right\|_1 \leqslant
\sqrt{\varepsilon (2 + \varepsilon + 2 \sqrt{1+\varepsilon})}.
\enq
Note that if it is known that $\xi^A \in \densitymatrix_\leqslant(\cH_A)$, then from Corollary 2.2 in
Ref. \cite{dupuis-2014}, the bound in the RHS can be refined to $2 \sqrt{\varepsilon}$.
\end{lemma}
\begin{proof}
We use the first inequality in the claim of Lemma \ref{yal14} to have
\beq
\tr \xi^A + \tr \Psi^A - 2 F(\xi^A,\Psi^A) \leqslant \varepsilon.
\enq
Using the Uhlmann's theorem \cite{uhlmann-1976},
we claim that there exists a partial isometry $V^{B \to C}$ such that
$F(\xi^A, \Psi^A)$ $= F(V^{B \to C} \cdot \xi^{AB}, \Psi^{AC})$, and hence,
\beq
\tr \xi^{AC} + \tr \Psi^{AC} - 2 F(V^{B \to C} \cdot \xi^{AB},\Psi^{AC}) \leqslant \varepsilon.
\enq
Since, $|\tr \xi^A - \tr \Psi^A| \leqslant \varepsilon$, or, $\tr \xi^A \leqslant 1 + 
\varepsilon$, and, using monotonicity, $F(V^{B \to C} \cdot \xi^{AB},\Psi^{AC}) 
\leqslant \sqrt{(\tr \xi^A) (\tr \Psi^A)}$ $\leqslant \sqrt{1+\varepsilon}$, and
hence,
$\tr \xi^{AC} + \tr \Psi^{AC} + 2 F(V^{B \to C} \cdot \xi^{AB},\Psi^{AC}) \leqslant 2 + \varepsilon + 2 \sqrt{1+\varepsilon}$.
Using the second inequality in the claim of Lemma \ref{yal14} again, we arrive at
\begin{multline}
\left\| V^{B \to C} \cdot \xi^{AB} - \Psi^{AC} \right\|_1 \\ \leqslant
\sqrt{ \left[ \tr \xi^{AC} + \tr \Psi^{AC} - 2 F(V^{B \to C} \cdot \xi^{AB},\Psi^{AC}) \right] \left[
\tr \xi^{AC} + \tr \Psi^{AC} + 2 F(V^{B \to C} \cdot \xi^{AB},\Psi^{AC}) \right] } \\
\leqslant \sqrt{\varepsilon (2 + \varepsilon + 2 \sqrt{1+\varepsilon})}.
\end{multline}
QED.
\end{proof}

\begin{corollary}[A straightforward corollary of Lemma 9.2 in Ref. \cite{hayashi}]
\label{yal10}
Consider a cq state
\beq
\rho^{XR} \equiv \sum_{x \in \cX} p_x \ketbra{x}^X \otimes \rho_x^R,
\enq
where $\rho_x^R \in \densitymatrix(\cH_R)$, $x \in \cX$, and $\{ p_x, x \in \cX \}$
is a probability vector. Let $\rho^R = \tr_X \rho^{XR}$, $\zeta > 0$,
$M \in \mathbb{N}$, any $\kappa^R \in \densitymatrix(\cH_R)$,
and $X^M \equiv (X_1,...,X_M)$ be $M$ i.i.d. random variables with probability distribution $\{p_x, x \in \cX\}$.
Then we have for any $\alpha \in (1,2]$,
\beq
\Exp_{X^M} \Big\| \frac{1}{M} \sum_{i=1}^M \rho_{X_i}^R - \rho^R \Big\|_1
\leqslant 4 \, \exp \left\{ \frac{\alpha-1}{2 \alpha} \left[ \log \nu_{\kappa^R} + D_\alpha(\rho^{XR} \| \rho^X \otimes
\kappa^R) - \log M \right] \right\}.
\enq
\end{corollary}
\begin{proof}
It follows from the claims of Lemma 9.2 in Ref. \cite{hayashi} that for any $\zeta > 0$,
\begin{align}
\Exp_{X^M} \Big\| \frac{1}{M} \sum_{i=1}^M \rho_{X_i}^R - \rho^R \Big\|_1 & \leq
2 \sum_x p_x \zeta^{\frac{1-\alpha}{2}} \sqrt{ Q_{\alpha}(\rho_x^R \| \kappa^R) } +
\sqrt{\frac{\nu_{\kappa^R} \zeta}{M}} \\
& = 2 \zeta^{\frac{1-\alpha}{2}} \exp\Big\{ \frac{\alpha-1}{2}
D_\alpha(\rho^{XR} \| \rho^X \otimes \kappa^R) \Big\} + \sqrt{\frac{\nu_{\kappa^R} \zeta}{M}}.
\end{align}
If we make a choice of
\beq
\zeta = \left( \frac{2 \exp\left\{ \frac{\alpha-1}{2} D_\alpha(\rho^{XR} \| \rho^X \otimes \kappa^R) \right\} M}
{\nu_{\kappa^R}} \right)^{\frac{2}{\alpha}},
\enq
we get
\begin{align}
\Exp_{X^M} \Big\| \frac{1}{M} \sum_{i=1}^M \rho_{X_i}^R - \rho^R \Big\|_1 & \leq
4 \, \exp \left\{ \frac{\alpha-1}{2 \alpha} \left[ \log \nu_{\kappa^R} + D_\alpha(\rho^{XR} \| \rho^X \otimes
\kappa^R) - \log M \right] \right\}.
\end{align}
QED.
\end{proof}

\section{A more general decoupling theorem that we never use!}
\label{appendix3}

\begin{theorem}
\label{theorem1-2}
Let $\cX$ be a finite set,
$\{ p_x, x \in \cX \}$ a probability distribution on $\cX$, $\rho_x^{AR} \in \densitymatrix(\cH_{AR})$
$\forall$ $x \in \cX$, and $\{ \ketbra{x}^X \}$ a set of orthonormal
states in $X$. Consider a cq state
\beq
\rho^{XAR} \equiv \sum_{x \in \cX} p_x \ketbra{x}^X \otimes \rho_x^{AR}.
\enq
For $M \in \mathbb{N}$, let $X_1,...,X_M$ be $M$ independent
and identically distributed (i.i.d.) random variables having probability distribution $\{p_x, x \in \cX\}$,
and $\cT^{A \to E}$ be a class-$1$ map. 
Then for $\alpha \in (1,2]$, $X_1^M \equiv (X_1,...,X_M)$, random Unitaries
$U_1^M \equiv (U_1,...,U_M)$ acting independently on $A$, we have for any $\sigma^R$,
$\kappa^R \in \densitymatrix(\cH_{R})$,
\begin{align}
\Exp_{X_1^M} \Exp_{U_1^M} & \Big\| \frac{1}{M} \sum_{i=1}^M \cT(U_i \cdot \rho^{AR}_{X_i}) -
\omega^E_\cT \otimes \rho^R \Big\|_1 \nonumber \\
& \leqslant 4 \exp\left\{ \frac{\alpha-1}{2\alpha} \left[ \log \nu_{\sigma^R}
+ D_\alpha(\rho^{XAR} \| \rho^X \otimes \eye^A \otimes \sigma^R) - \log M
+ \Theta(\cT) \right] \right\} \ind_{|A| \neq 1} \nonumber \\
& \hspace{0.5in} + 4 \exp \left\{ \frac{\alpha-1}{2\alpha} \left[ \log \nu_{\kappa^R} +
D_{\alpha}(\rho^{XR} \| \rho^X \otimes  \kappa^R) - \log M \right] \right\} \ind_{|\cX| \neq 1}.
\end{align}
\end{theorem}
\begin{proof}
We have
\begin{align}
\Exp_{X_1^M} & \Exp_{U_1^M} \Big\| \frac{1}{M} \sum_{i=1}^M \cT(U_i \cdot \rho^{AR}_{X_i}) -
\omega^E_\cT \otimes \rho^R \Big\|_1 \nonumber \\
%& = \Exp_{X_1^M} \Exp_{U_1^M} \Big\| \frac{1}{M} \sum_{i=1}^M \left[ \cT(U_i \cdot \rho^{AR}_{X_i})
%- \omega^E_\cT \otimes \rho^R_{X_i} + \omega^E_\cT \otimes \rho^R_{X_i} \right]
%- \omega^E_\cT \otimes \rho^R \Big\|_1 \nonumber \\
& \leqslant \Exp_{X_1^M} \Exp_{U_1^M} \Big\| \frac{1}{M} \sum_{i=1}^M \left[ \cT(U_i \cdot \rho^{AR}_{X_i})
- \omega^E_\cT \otimes \rho^R_{X_i} \right] \Big\|_1 + %\nonumber \\
%& \hspace{1in}
\Exp_{X_1^M} \Big\|_1 \frac{1}{M} \sum_{i=1}^M \omega^E_\cT \otimes
\rho^R_{X_i} - \omega^E_\cT \otimes \rho^R \Big\|_1 \nonumber \\
%& = \Exp_{X_1^M} \Exp_{U_1^M} \Big\| \frac{1}{M} \sum_{i=1}^M \left[ \cT(U_i \cdot \rho^{AR}_{X_i})
%- \omega^E_\cT \otimes \rho^R_{X_i} \right] \Big\|_1 +
%\Exp_{X_1^M} \Big\|_1 \frac{1}{M} \sum_{i=1}^M \rho^R_{X_i} - \rho^R \Big\|_1 \nonumber \\
\label{yae44}
& = \Exp_{X_1^M} \Exp_{U_1^M} \Big\| \frac{1}{M} \sum_{i=1}^M \left[ \cT(U_i \cdot \rho^{AR}_{X_i})
- \omega^E_\cT \otimes \rho^R_{X_i} \right] \Big\|_1 \ind_{|A| \neq 1} + \nonumber \\
& \hspace{1in}
\Exp_{X_1^M} \Big\|_1 \frac{1}{M} \sum_{i=1}^M \rho^R_{X_i} - \rho^R \Big\|_1 \ind_{|\cX| \neq 1},
\end{align}
where the inequality follows from the triangle inequality and the last equality follows since
$\| X \otimes Y \|_1 = \|X \|_1 \|Y \|_1$, and
the first and the second terms are identically zero if $|A|=1$ and $|\cX|=1$ respectively.
The upper bound for the second term can be deduced
from Lemma 9.2 in Ref. \cite{hayashi} for any $\alpha \in (1,2]$ 
and any $\kappa^R \in \densitymatrix(\cH_R)$ as
\beq
\label{yae42}
\Exp_{X_1^M} \Big\|_1 \frac{1}{M} \sum_{i=1}^M \rho^R_{X_i} - \rho^R \Big\|_1 \leq
4 \exp \left\{ \frac{\alpha-1}{2\alpha} \left[ \log \nu_R + D_{\alpha}(\rho^{XR} \| \rho^X \otimes \kappa^R)
- \log M \right] \right\}.
\enq
Note that Lemma 9.2 in Ref. \cite{hayashi} doesn't provide an upper bound in the above form but it
is easy to deduce it from the claim, and, for the sake of completeness,
it is provided in Corollary \ref{yal10}.

The rest of the proof is to upper bound the first term in \eqref{yae44}.
For $\zeta > 0$ and $\forall$ $x \in \cX$, let
$\Pi_x^{AR} \equiv \left\{ \cM_{\eye^A \otimes \sigma^R}(\rho_x^{AR}) \geqslant \zeta \eye^A \otimes \sigma^R
\right\}$, $\hat{\Pi}_{x}^{AR} \equiv \eye^{AR} - \Pi_{x}^{AR}$,
$\mu_{1,x} \equiv \omega^{E}_\cT \otimes \tr_A \left\{ \Pi_x^{AR} \rho^{AR}_x \right\}$, and
$\mu_{2,x} \equiv \omega^{E}_\cT \otimes \tr_A \left\{ \hat{\Pi}_x^{AR} \rho^{AR}_x \right\}$.
Note that $\mu_{1,x} + \mu_{2,x} = \omega^{E}_\cT \otimes \rho^R_x$. We now have
from the triangle inequality
\begin{align}
\Exp_{X_1^M} & \Exp_{U_1^M} \Big\| \frac{1}{M} \sum_{i=1}^M \left[ \cT(U_i \cdot \rho^{AR}_{X_i})
- \omega^E_\cT \otimes \rho^R_{X_i} \right] \Big\|_1 \nonumber \\
%& \hspace{0.1in}
%= \Exp_{X_1^M} \Exp_{U_1^M} \Big\| \frac{1}{M} \sum_{i=1}^M \left\{ \cT \left[ U_i \cdot (\Pi_{X_i}^{AR}
%\rho^{AR}_{X_i}) \right] - \mu_{1,X_i} +
%\cT \left[ U_i \cdot (\hat{\Pi}_{X_i}^{AR} \rho^{AR}_{X_i}) 
%\right] - \mu_{2,X_i} \right\} \Big\|_1 \\
& \leqslant \Exp_{X_1^M} \Exp_{U_1^M} \Big\| \frac{1}{M} \sum_{i=1}^M \left\{
\cT \left[ U_i \cdot (\Pi_{X_i}^{AR} \rho^{AR}_{X_i}) \right] - \mu_{1,X_i} \right\} \Big\|_1 + \nonumber \\
& \hspace{1in} \Exp_{X_1^M} \Exp_{U_1^M} \Big\|_1 \frac{1}{M} \sum_{i=1}^M \left\{ \cT \left[ U_i
\cdot (\hat{\Pi}_{X_i}^{AR} \rho^{AR}_{X_i}) \right] - \mu_{2,X_i} \right\} \Big\|_1.
\end{align}

We attack the first term.
\begin{align}
\Exp_{X_1^M} & \Exp_{U_1^M} \Big\| \frac{1}{M} \sum_{i=1}^M \left\{
\cT \left[ U_i \cdot (\Pi_{X_i}^{AR} \rho^{AR}_{X_i})  \right] - \mu_{1,X_i} \right\} \Big\|_1 \nonumber \\
& \leqslant \Exp_{X_1^M}  \Exp_{U_1^M} \Big\| \frac{1}{M} \sum_{i=1}^M \cT \left[ U_i \cdot (\Pi_{X_i}^{AR}
\rho^{AR}_{X_i})  \right] \Big\|_1 + \Exp_{X_1^M} \| \mu_{1,X_i} \|_1 \\
& \leqslant \frac{2}{M} \sum_{i=1}^M \Exp_{X_i}  \Exp_{U_i} \Big\| \cT \left[ U_i \cdot (\Pi_{X_i}^{AR}
\rho^{AR}_{X_i})  \right] \Big\|_1
= 2 \, \Exp_{X}  \Exp_{U} \Big\| \cT \left[ U \cdot (\Pi_{X}^{AR} \rho^{AR}_{X}) \right] \Big\|_1 \\
& \leqslant 2 \, \Exp_{X}  \Big\| \Pi_{X}^{AR} \rho^{AR}_{X} \Big\|_1
= 2 \, \sum_x p_x  \Big\| \Pi_{x}^{AR} \rho^{AR}_{x} \Big\|_1 \\
& \leqslant 2 \zeta^{\frac{1-\alpha}{2}} \, \sum_x p_x \sqrt{
Q_\alpha(\rho_x^{AR} \| \eye^A \otimes \sigma^R) } \\
& = 2 \zeta^{\frac{1-\alpha}{2}} \, \exp \left\{ \frac{\alpha - 1}{2}
D_\alpha(\rho^{XAR} \| \rho^X \otimes \eye^A \otimes \sigma^R) \right\},
\end{align}
where the first inequality follows from the triangle inequality,
the second inequality follows from the convexity of the trace norm to have
\begin{multline}
\Exp_{X_1^M} \| \mu_{1,X_i} \|_1 = \Exp_{X_1^M} \Big\| \frac{1}{M}
\sum_{i=1}^M  \Exp_{U_i}  \left\{ \cT \left[ U_i \cdot (\Pi_{X_i}^{AR}
\rho^{AR}_{X_i})  \right] \right\} \Big\|_1 \\
\leqslant \frac{1}{M} \sum_{i=1}^M \Exp_{X_i} \Big\| \Exp_{U_i}
\left\{ \cT \left[ U_i \cdot (\Pi_{X_i}^{AR} \rho^{AR}_{X_i})  \right] \right\} \Big\|_1
\leqslant \frac{1}{M} \sum_{i=1}^M \Exp_{X_i}  \Exp_{U_i} \Big\| \cT \left[ U_i \cdot (\Pi_{X_i}^{AR}
\rho^{AR}_{X_i})  \right] \Big\|_1,
\end{multline}
and similarly for the first term, the first equality follows since $X_i$'s and $U_i$'s are i.i.d.,
the third inequality follows from the definition of class-$1$ maps,
the fourth inequality follows from Lemma \ref{yal3} (proved by Hayashi \cite{hayashi}),
and the last inequality follows from the concavity of $x \mapsto \sqrt{x}$.

We now attack the second term. Let
$\Delta_{X_i U_i} \equiv \cT \left[ U_i \cdot (\hat{\Pi}_{X_i}^{AR} \rho^{AR}_{X_i}) \right] - \mu_{2,X_i}$ and
$\Delta_{X_1^M U_1^M} \equiv \sum_{i=1}^M \Delta_{X_i U_i}/M$.
Note that
$\Exp_{X_1^M U_1^M} \Big\{ \Delta_{X_i U_i} \Delta_{X_j U_j}^\dagger \Big\} =
\mzero$, $\forall$ $i \neq j$,
and hence,
\begin{align}
\Exp_{X_1^M U_1^M} \Big\{ \Delta_{X_1^M U_1^M} \Delta_{X_1^M U_1^M}^\dagger \Big\}
& = \frac{1}{M^2} \sum_{i=1}^M \Exp_{X_i U_i} \Big\{ \Delta_{X_i U_i} \Delta_{X_i U_i}^\dagger
\Big\}
= \frac{1}{M} \Exp_{X U} \Big\{ \Delta_{X U} \Delta_{X U}^\dagger \Big\} \\
\label{yae16}
& \leqslant \frac{|A|^2 \tr_{A^\prime} \left( \omega^{E A^\prime}_\cT \right)^2 }{M(|A|^2-1)} \otimes
\tr_A \Exp_X \left\{ \hat{\Pi}_{X}^{AR} (\rho^{AR}_{X})^2 \hat{\Pi}_{X}^{AR} \right\},
\end{align}
where the inequality follows from Lemma \ref{yal9}. Following the arguments in Theorem \ref{theorem1}
in dealing with the second term, we get
\begin{align}
\Exp_{X_1^M U_1^M} & \Big\| \frac{1}{M} \sum_{i=1}^M \left\{ \cT \left[ U_i \cdot (\hat{\Pi}_{X_i}^{AR}
\rho^{AR}_{X_i}) \right] - \mu_{2, X_i} \right\} \Big\|_1 \leqslant
\sqrt{ \frac{\nu_{\sigma^R} \zeta |A|^2 \exp \left\{ \Theta(\cT) \right\}}{M (|A|^2-1)} }.
\end{align}

We now have
\begin{multline}
\Exp_{X_1^M} \Exp_{U_1^M} \Big\| \frac{1}{M} \sum_{i=1}^M \left[ \cT(U_i \cdot \rho^{AR}_{X_i})
- \omega^E_\cT \otimes \rho^R_{X_i} \right] \Big\|_1 \\
\leqslant 2 \zeta^{\frac{1-\alpha}{2}} \, \exp \left\{ \frac{\alpha - 1}{2} \left[
D_\alpha(\rho^{XAR} \| \rho^X \otimes \eye^A \otimes \sigma^R) \right] \right\}
+ \sqrt{ \frac{\nu_{\sigma^R} \zeta |A|^2 \exp \left\{ \Theta(\cT) \right\}}{M (|A|^2-1)} },
\end{multline}
and by appropriately choosing $\zeta$, we get
\begin{multline}
\label{yae54}
\Exp_{X_1^M} \Exp_{U_1^M} \Big\| \frac{1}{M} \sum_{i=1}^M \left[ \cT(U_i \cdot \rho^{AR}_{X_i})
- \omega^E_\cT \otimes \rho^R_{X_i} \right] \Big\|_1 \\
\leqslant 4 \exp\left\{ \frac{\alpha-1}{2\alpha} \left[ \log \nu_{\sigma^R}
+ D_\alpha(\rho^{XAR} \| \rho^X \otimes \eye^A \otimes \sigma^R) - \log M
+ \Theta(\cT) \right] \right\}.
\end{multline}
The claim now follows from \eqref{yae44}, \eqref{yae42}, and  \eqref{yae54}.
\end{proof}

%\bibliographystyle{myfile}
%\bibliography{master}
%

\end{document}